\documentclass[
times,
aps,
twocolumn,
showpacs,
superscriptaddress,
pre,
nofootinbib,
floatfix]{revtex4-1}

\usepackage{graphicx}
\usepackage{array}
\usepackage{color}
\usepackage{upgreek}
\usepackage{float}
\usepackage{pifont}
\usepackage{enumerate}
\usepackage{enumitem}
\usepackage{lipsum}
\usepackage{booktabs}
\usepackage{url}
\usepackage{braket}

\usepackage{amsthm}
\usepackage{bm}
\usepackage[T1]{fontenc}
\usepackage{scrextend}
\usepackage{hyperref}
\usepackage[dvipsnames, table]{xcolor}

\definecolor{BlueRome}{HTML}{4287f5}
\definecolor{C1}{RGB}{52, 89, 149}
\definecolor{C2}{RGB}{251, 77, 61}
\definecolor{C3}{RGB}{3, 206, 164}
\definecolor{C4}{RGB}{202, 21, 81}
\definecolor{cyanRed}{RGB}{252, 5, 120}
\hypersetup{colorlinks=true, linkcolor=BlueRome, citecolor=BlueRome, urlcolor=cyanRed}

\usepackage{thm-restate}

\usepackage{dsfont}
\usepackage{epstopdf}
\usepackage{tikz}
\usetikzlibrary{fadings}
\usepackage{mathrsfs}
\usepackage[export]{adjustbox}
\usepackage{thmtools}
\usepackage{thm-restate}

\usepackage{fdsymbol}

\usepackage{accents}

\newtheorem{thm}{Theorem}

\newtheorem{lemma}[thm]{Lemma}

\newtheorem{crllr}[thm]{Corollary}

\theoremstyle{remark}

\newcommand*{\nn}{\nonumber}

\newcommand*{\id}{\mathds{1}}

\newcommand*{\mc}{\mathcal}
\newcommand*{\dg}{\dagger}
\newcommand*{\ex}{\mathrm{e}}

\DeclareMathOperator{\tr}{tr}

\newcommand*{\drho}{\psi}

\newcommand{\x}{\boldsymbol{x}}
\newcommand{\w}{\boldsymbol{w}}
\newcommand{\WU}[1]{W_{#1}}

\newcommand{\Wopt}{W_{\mathrm{univ}}}
\newcommand{\mbu}{\mathbb{U}}

\newcommand*{\mce}{\mathcal{E}}

\hypersetup{
pdftitle={ChaosQML},
pdfsubject={Many-body quantum chaos},
pdfauthor={Neil Dowling, Maxwell T. West, Muhammad Usman, and Kavan Modi},
pdfkeywords={Quantum chaos, Many-body quantum physics, OTOC, Quantum Machine Learning}
}

\begin{document}
\title[]{Adversarial Robustness Guarantees for Quantum Classifiers}

\author{Neil Dowling}
\thanks{Equal contribution authors, listed in alphabetical order.}
\affiliation{Institut f\"ur Theoretische Physik, Universit\"at zu K\"oln, Z\"ulpicher Strasse 77, 50937 K\"oln, Germany}
\affiliation{School of Physics \& Astronomy, Monash University, Clayton, VIC 3800, Australia}

\author{Maxwell T. West}
\thanks{Equal contribution authors, listed in alphabetical order.}
\affiliation{School of Physics, The University of Melbourne, Parkville, VIC 3010, Australia}

\author{Angus Southwell}
\affiliation{School of Physics \& Astronomy, Monash University, Clayton, VIC 3800, Australia}

\author{Azar C. Nakhl} 
\affiliation{School of Physics, The University of Melbourne, Parkville, VIC 3010, Australia}

\author{Martin Sevior} 
\affiliation{School of Physics, The University of Melbourne, Parkville, VIC 3010, Australia}

\author{Muhammad Usman}
\affiliation{School of Physics, The University of Melbourne, Parkville, VIC 3010, Australia} \affiliation{Data61, CSIRO, Clayton, 3168, VIC, Australia}

\author{Kavan Modi}
\affiliation{Science, Mathematics and Technology Cluster, Singapore University of Technology and Design, \\8 Somapah Road, 487372 Singapore}
\address{School of Physics \& Astronomy, Monash University, Clayton, VIC 3800, Australia}


\begin{abstract}
Despite their ever more widespread deployment throughout society, machine learning algorithms remain critically vulnerable to being spoofed by subtle adversarial tampering with their input data. The prospect of near-term quantum computers being capable of running {quantum machine learning} (QML) algorithms has therefore generated intense interest in their adversarial vulnerability. Here we show that quantum properties of QML algorithms can confer fundamental protections against such attacks, in certain scenarios guaranteeing robustness against classically-armed adversaries. We leverage tools from many-body physics to identify the quantum sources of this protection. Our results offer a theoretical underpinning of recent evidence which suggest quantum advantages in the search for adversarial robustness.
In particular, we prove that quantum classifiers are: (i) protected against weak perturbations of data drawn from the trained distribution, (ii) protected against local attacks if they are insufficiently scrambling, and (iii) show evidence that they are protected against universal adversarial attacks if they are sufficiently chaotic. Our analytic results are supported by numerical evidence demonstrating the applicability of our theorems and the resulting robustness of a quantum classifier in practice. This line of inquiry constitutes a concrete pathway to advantage in QML, orthogonal to the usually sought improvements in model speed or accuracy.
\end{abstract}


\maketitle

\section{Introduction}
Ten years on from their initial discovery~\cite{szegedy2013intriguing,biggio2013evasion,goodfellow2014explaining}, adversarial attacks remain a potent weapon for deceiving even highly sophisticated machine learning (ML) models~\cite{chakraborty2021survey}. 
Remarkably, for example, powerful image classifiers can be fooled by carefully chosen perturbations which are almost invisible to a human eye~\cite{ilyas2019adversarial}, or even by changing the value of a single pixel~\cite{su2019one}.
Due to the accelerating delegation of important tasks to ML, and the tendency of empirical defense strategies to be later bypassed~\cite{athalye2018obfuscated}, the need for provable guarantees against such spoofing attempts is only growing~\cite{cohen2019certified,lecuyer2019certified}.

Concurrently, the increasing capabilities of quantum computers have generated significant research to determine whether quantum advantage may be expected in machine learning~\cite{biamonte2017quantum,abbas2021power,liu2021rigorous,holmes2022connecting}, but the extent to which they can be expected to deliver direct speed-ups remains unclear~\cite{mcclean2018barren,wang2021noise,holmes2022connecting,cerezo2021cost,larocca2022diagnosing,patti2021entanglement,ragone2023unified,diaz2023showcasing,fontana2023adjoint,larocca2024review,cerezo2023does}. 
It is therefore an opportune moment to search for a different kind of advantage in QML~\cite{west2023towards,NMI2023}. In fact, 
the field of quantum adversarial machine learning has generated considerable interest~\cite{lu2020quantum,liu2020vulnerability,du2021quantum,guan2021robustness,weber2021optimal,liao2021robust,kehoe2021defence,ren2022experimental,west2023towards,wu2023radio,west2023benchmarking,west2023drastic,khatun2024quantum,winderl2024constructing,berberich2023training}. Notably, in a series of recent papers, QML models were studied that indicated significantly increased adversarial robustness against classical adversaries~\cite{west2023benchmarking,west2023drastic,wu2023radio,khatun2024quantum} (Fig.~\ref{fig:1}(a)). However, these results are empirical, lacking a foundational understanding of the source of the advantage.

In this work we address this by supplying a sequence of provable quantum adversarial robustness guarantees for QML, 
in extremely broad yet practically relevant scenarios.
These rely on distinct properties of the encoding scheme, as well as on the dynamical complexity of the constituent quantum circuit. Our results include analytic theorems relying on the genuinely quantum properties of a QML architecture, offering robustness guarantees not applicable to classical ML.
These are further supported with probabilistic bounds and numerical results for a realistic quantum classifier model. 
These guarantees circumvent previous existence proofs of adversarial examples in QML~\cite{liu2020vulnerability,liao2021robust}, by restricting to the physically relevant case of a classical adversary whose allowable perturbations are constrained by the data encoding strategy employed by the model. More specifically, we study the robustness of QML models under three distinct attack scenarios: a weak perturbation designed to induce a misclassification for a target input classical state (data), a strong \textit{universal} perturbation~\cite{hendrik2017universal,gong2022universal} designed to induce misclassification on all states, and a strong local perturbation targeting a specific input state. 
Our results are summarised in Table~\ref{tab:results}.

\begin{figure}
  \includegraphics[width=0.48\textwidth]{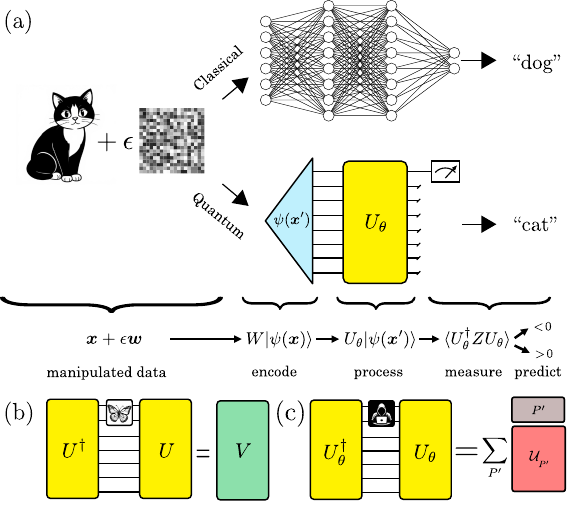}
  \caption{\textbf{Schematic of adversarial machine learning setting.} (a) Machine learning models are generally highly susceptible to extremely subtle adversarial tampering with their input data, but quantum models have been empirically found to be robust to attacks by classical adversaries~\cite{west2023benchmarking}. In the general quantum machine learning setting, a classical data string $\x$ is encoded in a state $\ket{\psi(\x)}$, a (trained) quantum algorithm $U_{\theta}$ is applied before measurement of some few-qubit operator $Z$. An adversarial attack can then be modeled by some change to the initial bit string $\x \to \x + \epsilon \w$, which is equivalent to the action of a unitary $W$ on the encoded state, $\ket{\x'} =  W\ket{\x} $. (b) Chaotic unitaries scramble information throughout quantum degrees of freedom in a many-body system. (c) It is difficult for an adversary to carefully manipulate a chaotic circuit in the precise way needed to induce misclassification. Here, $P'$ is some (spoofed) Pauli string which flips the measurement outcome, and $\mathcal{U}_{P'}$ is some unitary on the subsystem which is not measured. }
\label{fig:1}
\end{figure}

Our results are split into two categories. The first of these leverage the unitarity of quantum circuits, together with the atypical nature of training-set data in ML.

For the second set of results, techniques from the theory of many-body quantum chaos play a key role~\cite{Prosen2007,Prosen2007a,Prosen2009,Shenker_Stanford_2014,Maldacena_Shenker_Stanford_2016,Swingle2016,Roberts2016,Dubail_2017,Swingle2018,Jonay2018,Alba2019,Kos2020,Alba2021, Anand2021-yi,Rosa2022,dowling_scrambling_2023,dowling_operational_2024}. Viewing a variational algorithm as a many-body quantum system,
a trained circuit can be viewed as dynamics, with its depth serving as a proxy for a time parametrisation. In light of this, a variational circuit can be ascribed relevant properties from the field of many-body theory, such as whether it can effectively scramble information, whether certain observables from it can be classically simulatable, or whether it can be classed as non-integrable. In this work we derive a series of theorems dictating how these many-body properties can supply protection against adversaries. 
 
Our first such result is that a variational circuit needs to be scrambling for a local adversarial attack to be possible. This is measured by a quantity called the out-of-time-ordered correlator (OTOC), which dictates how information from an initially local operator spreads when it evolves according to the Heisenberg picture~\cite{Sekino_Susskind_2008,Shenker_Stanford_2014, Maldacena_Shenker_Stanford_2016, Swingle2016, Roberts2016,Swingle2018, Foini2019, PhysRevLett.124.140602, dowling_scrambling_2023}. The scrambling nature of explicit (backwards-applied) QML architectures has previously been observed~\cite{Shen2020,We2021}, suggesting that such circuits could be vulnerable to attack. However, in stark contrast to the necessity of scrambling, we will argue that a circuit which exhibits quantum chaos is, in fact, protected against universal adversaries. By `chaos' we mean the linear scaling of local-operator entanglement (LOE)~\cite{Prosen2007,Prosen2007a}. This operator quantity measures the complexity of simulating a Heisenberg operator as a matrix product operator, and in locally-interacting many-body systems, it has been observed to grow maximally (linearly) with time only for non-integrable dynamics~\cite{Prosen2009,Dubail_2017,Jonay2018,Alba2019,Kos2020,Alba2021,dowling_scrambling_2023}. This is a strictly stronger condition on the circuit compared to scrambling~\cite{dowling_scrambling_2023}. Therefore, in the sense of operator entanglement, a high quantum complexity of a trained variational algorithm has a fundamental robustness to adversarial attacks. 
Intuitively, a chaotic unitary will uncontrollably distribute a perturbation throughout the system (Fig.~\ref{fig:1}(b)), making it impossible to apply the specific adversarial manipulation needed to spoof the model (Fig.~\ref{fig:1}(c)).
Before detailing these results, we first describe the general setup and necessary background knowledge in adversarial QML.

The general QML models we consider in this work follow a standard three-step architecture consisting respectively of data encoding, data processing and measurement. 

In the first step, a vector $\x$ of classical data is loaded into the quantum computer by means of some encoding method,
\begin{equation}\label{eq:encodingArb}
 \mathcal{E}(\x) \ket{\textbf{0}} = \ket{\psi(\x)},
\end{equation}
where $\ket{\textbf{0}}:=\ket{0}^{\otimes n}$, with the number of qubits $n = n_{\mc{E}}(N)$ being some encoding-dependent function of the size of the classical data, $N:=\mathrm{len}(\x)$. It will often be necessary to consider the corresponding density matrix, $\psi(\x):=\ket{\psi(\x)}\bra{\psi(\x)}$.
We will explore three of the most natural and popular encoding methods~\cite{larose2020robust}, amplitude encoding, angle encoding, and dense encoding:
\begin{align}
&\mathcal{E}_{\mathrm{amp}}(\x)\ket{\textbf{0}} = \frac{1}{\|\x\|_2}\sum_{j=1}^N x_j \ket{j}, \label{eq:encodingAmp}\\
 &\mathcal{E}_{\mathrm{angle}}(\x)\ket{\textbf{0}} = \bigotimes_{j=1}^N e^{-i x_j \sigma_y} \ket{\textbf{0}},    \label{eq:encodingAngle}\\
 &\mathcal{E}_{\mathrm{dense}}(\x) \ket{\textbf{0}}= \bigotimes_{j=1}^{N/2} e^{-i x_{2j} \sigma_z} e^{-i x_{2j-1} \sigma_y} \ket{\textbf{0}}.\label{eq:encodingDense}
\end{align}
Each of these encodings has different advantages, such as in terms of expressibility,  learnability  and resource costs~\cite{larose2020robust}. For instance, amplitude encoding is drastically more space efficient, requiring a number of qubits only logarithmic in the dimension of the data, but in general requires exponentially deep circuits~\cite{larose2020robust,west2023drastic}, as opposed to the constant depth yet size-inefficient angle encoding circuits.

In the second and third steps, the encoded state is acted upon by a trainable variational unitary $U_{\theta}$, following which a local measurement is made. Without loss of generality, we choose this measurement to be of the Pauli-$Z$ operator on the first $k$ qubits. 
In the {binary classification case}, upon which for simplicity we focus (although the generalisation to multiple classes  is straightforward) 
we take the prediction  $\hat{y}$ of the model on the input $\x$ to be the sign of the final measurement,  
\begin{equation}
{y}_{\theta}(\x) := \bra{\psi(\x)}{U}_{\theta}^\dagger Z {U}_{\theta}\ket{\psi(\x)},\label{eq:prediction}
\end{equation}
choosing $Z:=\sigma_z^{\otimes k}\otimes{I}^{\otimes(n-k)}$ as the operator to be measured. We will also use the notation $Z_U:=U_\theta^\dg ZU_\theta$ to denote an operator Heisenberg-evolved by the circuit.
During training, the parameters $\theta$ are optimised to minimise a loss function $\ell_{\theta}(\x,y)$, for example of the form
\begin{equation}
 \ell_{\theta}(\x,y) = -y(\x)y_{\theta}(\x)\label{eq:loss}
\end{equation}
over a training set $S$, where we denote the true label of the datapoint $\x$ as $y(\x)\in\{\pm 1\}$.

An \textit{adversarial attack} is a vector $\w$ which perturbs the input data as $\x\mapsto \x' := \x+\w$, intended to change the prediction of the model. At the quantum circuit level, after encoding the classical perturbation induces a
unitary $W$ satisfying $\ket{\psi(\x')} := W \ket{\psi(\x)}$. An attack is deemed successful if this perturbation changes the prediction of the model, 
\begin{equation}
\mathrm{sgn}[{y}_{\theta}(\x')]= -\mathrm{sgn}[{y}_{\theta}(\x)]. \label{eq:attack}
\end{equation}
The properties of $W$ will be heavily influenced by the choice of data encoding map. In the case of angle or dense encoding (Eq.~\eqref{eq:encodingAngle}, ~\eqref{eq:encodingDense}) for example, it will take the form of a product of local unitaries $W=\bigotimes_i W_i$, but this is not generally true for amplitude encoding (Eq.~\eqref{eq:encodingAmp}).

\begin{table}
\begin{tabular}{c|c|c|c|}
\cline{2-4}
  & \begin{tabular}{@{}c@{}}Weak \\ (Thm.~\ref{thm:unitary})\end{tabular}   &  \begin{tabular}{@{}c@{}}Local \\ (Thm.~\ref{thm:unitary} \&~\ref{prop:otoc}) \end{tabular} & \multicolumn{1}{c|}{\cellcolor{white}\begin{tabular}{@{}c@{}}Universal \\ (Thm.~\ref{thm:universal}) \end{tabular}}   \\  \hline 
\multicolumn{1}{|c|}{Amplitude} &  $\checkmark$  & $\checkmark$  &   -- \\ \hline
\multicolumn{1}{|c|}{Angle}  & $\epsilon \lesssim 1/\sqrt{N}$   &   & $\checkmark$  \\  \cline{1-2} \cline{4-4}
\multicolumn{1}{|c|}{Dense}  & $\epsilon \lesssim 1/\sqrt{N}$   &  OTOC $\ll 1$ &  $\checkmark$ \\ \cline{1-2} \cline{4-4}
\multicolumn{1}{|c|}{Arbitrary} & $\epsilon \lesssim \left\lvert \Delta \boldsymbol{x} / \Delta \drho\right\rvert $ &    & --  \\ \hline
 & \begin{tabular}{@{}c@{}}Quantum\end{tabular}   &  \begin{tabular}{@{}c@{}}Scrambling  \end{tabular} & \multicolumn{1}{c|}{\cellcolor{white}\begin{tabular}{@{}c@{}}Chaotic\end{tabular}}   \\  \cline{2-4}
\end{tabular}
\caption{\textbf{Summary of robustness guarantees.} The applicability of our theorems, which depend on both the attack strategy and the form of data encoding, $\x\in\mathbb{R}^N\mapsto\drho(\x)=\mc{E}(\x)\ket{\boldsymbol{0}}\bra{\boldsymbol{0}}\mc{E}^\dg(\x)$. $\epsilon$ denotes the $
\ell_\infty$ norm of the adversarial perturbation. In some cases, our results apply unconditionally (denoted by a tick) while in others there is a specified dependence on the details of the encoding. Non-applicability is denoted by a dash. In the bottom row, we record the property of the model (qualitatively) responsible for the guarantee: ``Quantum'' refers to the contractive nature of any quantum classifier (e.g. a unitary circuit), ``Scrambling'' refers to a quickly decaying out-of-time-ordered correlator (OTOC) [Eq.~\eqref{eq:otocdef}], while by ``Chaotic'' we mean a linearly-growing local-operator entanglement (LOE) [Eq.~\eqref{eq:LOEdef}]. We also note that the ticks in the right-hand column are based on a conjecture, supported by numerical evidence and analytic results under a stronger condition than the most general universal adversarial attack (see Eqs.~\eqref{eq:uni_spoof} and \eqref{eq:anticomm}).}
\label{tab:results}
\end{table}

The existence of adversarial attacks which can spoof QML models (equivalently, the existence of nearby pairs of states classified differently) has already been established~\cite{liu2020vulnerability, liao2021robust}, and seems to indicate a significant vulnerability of quantum classifiers. What is less well understood, however, and the focus of this work, is the extent to which it is possible to construct and implement these attacks in practice. A key contribution of this work is to examine this adversarial setting in terms of concrete scenarios, identifying distinct (dis)advantages of the various encoding methods described above under different types of adversarial attack (see Table~\ref{tab:results}). We will derive fundamental robustness guarantees in three distinct situations: (i) tailored but weak attacks, (ii) strong, local attacks, and (iii) universal attacks which spoof all images with a single attack $\w$. To derive these guarantees, we will leverage the contractiveness of quantum dynamics, the dynamical complexity of the trained circuit $U_{\theta}$, and nature of the encoding in (Eqs.~\eqref{eq:encodingAngle}-\eqref{eq:encodingDense}). Remarkably, we will argue that given an encoding, and broad class of attack, our theorems hold in full generality. This means that they will apply to all future quantum adversarial scenarios that fit within one of these settings.

\section{Results}
At a high level our results (summarised in Table~\ref{tab:results}) can be split into two categories: statements about the strength of the perturbation required to induce a misclassification,
relying only on the unitarity/contractiveness the model (the first column of Table~\ref{tab:results}), and statements about the impossibility of carrying out certain classes of attacks, regardless of the strength of the perturbation (the second and third columns of Table~\ref{tab:results}). While in the first case we show a robustness for all quantum classifiers, the latter category relies on some notions from the theory of many-body chaos. We present our results based on progressively strongly requirements on the dynamical properties of the trained circuit $U_\theta$, as summarised in the bottom row of Table~\ref{tab:results}. 

\subsection{Weak Attacks} 
We first consider the simplest, and arguably the most potentially damaging, threat model: an input specific perturbation as weak as possible, so as to maximise the probability that the tampering is not detected. In this case, and in contrast to the generally highly non-linear nature of classical neural networks~\cite{agarap2018deep}, we can use the unitarity (and hence linearity) of isolated quantum circuits to arrive at the following result.
\begin{restatable}{thm}{unitary} \label{thm:unitary}
   Given an input state $\ket{\psi(\x)}$, a quantum model as defined in Eq.~\eqref{eq:prediction} will classify all states within a 1-norm ball of $\ket{\psi(\x)}$ of radius $ |{y}_{\theta}(\x)|$ identically.
\end{restatable}
A proof of Thm.~\ref{thm:unitary} may be found in Section~\ref{ap:weak_proof}.
In fact, this result extends beyond unitary circuits $U_\theta$ to entirely general quantum dynamics. This means, for example, that this robustness guarantee holds in the presence of noise (which effectively limits the adversary's control of the situation; e.g. see Ref.~\cite{du2021quantum}). 
The practical usefulness of Thm.~\ref{thm:unitary} depends on two factors: the extent to which changes in the classical data vector $\x$ translate to 1-norm changes in $\psi(\x)$, and how big one can expect the magnitude  $ |{y}_{\theta}(\x)|$ of the unperturbed output to be. 

Investigating the first point, which will depend on the choice of data encoding scheme, in Section~\ref{ap:weak_proof} we estimate $\Delta\drho=\| \psi(\x) - \psi(\x') \|_1$ under an adversarial attack $\x\mapsto\x'$ with $\Delta \x = \max_i |x_i - x'_i| = \epsilon\ll 1$ for our considered encoding schemes (Eqs.~\eqref{eq:encodingAmp}-\eqref{eq:encodingDense}). In the archetypal example of image classification, this would correspond to changing each pixel value by no more than $\epsilon$.
We find (see Section~\ref{ap:weak_proof}), for a classical data vector of length $N$,
\begin{align}
&    \Delta \drho_{\mathrm{angle}}\sim \sqrt{N}\epsilon + \mathcal{O}\left(\epsilon^2\right),\\
&    \Delta \drho_{\mathrm{dense}}\sim \sqrt{N}\epsilon + \mathcal{O}\left(\epsilon^2\right),\\
&    \Delta \drho_{\mathrm{amp}}\sim \epsilon + \mathcal{O}\left(\epsilon^2\right).
\end{align}
So for large $N$, when using angle or dense encoding changing each value of the input vector by a small amount can induce a large change in the corresponding quantum states, and effectively weakens the applicability of Thm.~\ref{thm:unitary} to perturbations with $|\epsilon|\lesssim 1/\sqrt{N}$. In the case of amplitude encoding, on the other hand, the resulting change is independent of $N$, implying that weakly perturbed data will be mapped close to the original irrespective of the dimension of the classical input data. This will therefore impart a strong robustness guarantee on the quantum classifier if $|{y}_{\theta}(\x)|$ is of appreciable magnitude. A similar analysis could be carried out for other encoding strategies; in general, the relevant quantity is the magnitude $\left\lvert \Delta \x / \Delta \drho \right\rvert$ of the change in the encoded state as a function of the change in the input classical data vector (see Table~\ref{tab:results}).

At first glance, however, it is unclear that one should expect $|{y}_{\theta}(\x)|$ to be reasonably large, as due to standard concentration effects the measurement values for a large quantum circuit will concentrate strongly around zero, and thus a small perturbation will be sufficient to change the sign of most (typical) measurement results, independent of the data encoding strategy. Indeed, this phenomenon has been used to conclude that in general quantum classifiers will possess extreme adversarial vulnerability, with perturbations of magnitude $\mathcal{O}(2^{-n})$, exponentially falling with the number of qubits, capable of changing the prediction of a model~\cite{liu2020vulnerability}.
However, and as has previously been recognised~\cite{liao2021robust}, the distribution of states in which one is interested in practice in ML is typically highly non-uniform, which can lead to a merely polynomially vanishing minimum perturbation size for a successful adversarial attack. Yet further progress can be made, eliminating entirely the dependence on the number of qubits, if one assumes that (modulo a potential adversarial perturbation) the test time samples are being drawn from the same distribution $\mu$ as was the training set $S$ over which the loss function $\ell_{\theta}$ (e.g. Eq.~\eqref{eq:loss}) was minimised. In this case, one has with high probability that the difference between the expected risk ${R}(\boldsymbol{\theta}) = \mathop{\mathbb{E}}_{(x,y)\sim \mu} \ell_{\theta} (x,y) $ and the empirical risk $\hat{R}_S(\boldsymbol{\theta}) = \frac{1}{|S|} \sum_{(x_i,y_i)\in S} \ell_{\theta} (x_i,y_i)$ is bounded by $\mathcal{O}\left(  \sqrt{\frac{T\log T}{|S|}}  \right) $ where $T$ is the number of trainable 2-qubit unitaries~\cite{caro2022generalization}. So if one trains until (say) $\hat{R}_S(\boldsymbol{\theta})<-1/2$ on a training set $S$ with $|S|\gtrsim T\log T$ then (with high probability) attacked states with {$\| \psi(\x)- \psi(\x') \|_1 < 1/4$ and $(\x,y(\x))\sim\mu$} will be classified correctly, requiring an adversary to implement perturbations of magnitude $\mathcal{O}(1)$. We conclude that a well-trained QML model will be drastically more adversarially robust on encoded states of data drawn from $\mu$ than it will be on (Haar) random states, which, conveniently, are exactly the states that we care about the most. 

A useful consequence of the above construction is that it also automatically implies robustness to non-adversarial noise. 
In particular, Thm.~\ref{thm:unitary} immediately also includes a robustness in terms of the strength $\epsilon$ of the noise in \textit{any} perturbation $\x \mapsto \x + \epsilon \w $. An adversarial perturbation can be seen as a ``worst case'' scenario, where $\w$ is picked to optimise changing the prediction in Eq.~\eqref{eq:prediction}, compared to noise where $\w$ would be sampled from some distribution. This can also be viewed as encompassing noise in the quantum algorithm itself. That is, for sufficient training, any weak, coherent noise will not rotate the state vector outside of the 1-norm ball of Thm.~\ref{thm:unitary}. Moreover, Thm.~\ref{thm:unitary} readily extends also to incoherent noise; see Section~\ref{ap:weak_proof}.

We note that in Ref.~\cite{berberich2023training}, in a conceptually similar argument to Thm.~\ref{thm:unitary}, Lipschitz bounds are employed to show how certain variational circuits can be trained to constrain (in our notation)  $|{y}_{\theta}(\x')-{y}_{\theta}(\x)|$ as a function of $\|\x'-\x\|$, and the resulting relation between expressivity and robustness is explored. The present work, however, relaxes an assumption made in Ref.~\cite{berberich2023training} on the form of data encoding, with Thm.~\ref{thm:unitary} being encoding-agnostic.

The results discussed so far depend on the strength of the attack being weak -- the isometric nature of unitary maps does not give as useful a bound when the attack is strong.
In the search for guarantees even in the face of strong perturbations, then, we need to turn to a different property of the models.
As we will see in the next two sections, covering local and universal attacks, it will turn out that different degrees of dynamical complexity in $U_\theta$ can either safeguard or jeopardise its robustness against strong adversarial attacks.

\subsection{Local Attacks} 
The second attack scenario that we consider is the case of strong, local attacks, where the assumption we make on $W$ is that it acts only on a few qubits. That is, instead of a weak attack in Thm.~\ref{thm:unitary}, we take $\x' = \x + \boldsymbol{w}$, where $w_i \neq 0$ only for some small number $k \ll N $ bits. This threat model is inspired by the surprising result that certain classical neural networks can be successfully spoofed even if the attacker can change only a single pixel~\cite{su2019one}. 

As a first result, we note that a strongly (but still locally) perturbed classical vector $\x'$ need not lead to strongly perturbed state $\ket{\psi(\x')}$ after encoding. In particular, in the case of amplitude encoding \eqref{eq:encodingAmp}, changing only a small number of bits necessarily leads to a weak attack, $\|\psi(\x) - \psi(\x') \|_1 \ll 1$. Intuitively, this is because for a large quantum state with many non-zero amplitudes, changing only a small fraction of these cannot significantly change the global state. The results of Thm.~\ref{thm:unitary} then directly apply, as we prove in Section~\ref{ap:local}. Therefore, for amplitude encoding, QML circuits are robust against local attacks in addition to weak attacks (first cell of the middle column in Table~\ref{tab:results}). For other forms of encoding, however, this will not be true in general. In the cases of angle and dense encoding, for example, one can orthogonalise a pair of encoded states by changing only a single element of the corresponding classical data vector. Nonetheless, we can make progress by considering the scrambling characteristic of $U_{\theta}$. 

In contrast to the previous results, we now make the simplifying assumption that the initial state is sampled from a unitary 2-design, a mild relaxation of the approach of e.g. Ref.~\cite{liu2020vulnerability}. This is admittedly a strengthening of the previous assumptions that will not hold in general; when it does, however, we have:
\begin{restatable}{thm}{localOTOC}
\label{prop:otoc}
    If the state $\ket{\psi(\x)}$ representing the classical data $\x$ is sampled from a 2-design $\mce$ over the uniform measure on states, then for any attacked state $\ket{\psi(\x')}=W\ket{\psi(\x)}$ and  for any $\delta > 0$,
    \begin{equation} \label{eq:otocLevy}
        \Pr_{\ket{\psi(\x)} \sim \mce}\left\{ |y_{\theta}(\x) - y_{\theta}(\x') |  \geq \delta \right\} \leq \frac{\braket{ [Z_{U},W ]^2 }}{(d+1)\delta^2},
    \end{equation}
    where the prediction ${y}_{\theta}(\x)$ is defined in Eq.~\eqref{eq:prediction}, and the expectation value on the r.h.s. is over a maximally mixed state, $\braket{ [Z_{U},W ]^2 } = (1/d) \tr([Z_{U},W ]^2 )$.
\end{restatable}
This theorem is entirely independent of the form of the attack and holds generally for any $W$. Similar encoding-specific bounds could be obtained by instead averaging over states attainable from a specific encoding scheme. We explain this in Section~\ref{ap:local}, and supply there a proof of the above theorem. When $W$ is a local attack, the numerator of the r.h.s. can be interpreted as an out-of-time-order correlator (OTOC), which probes how scrambling  the process $U_\theta$ is. We first interpret the theorem for a generic $W$ and then analyze it in terms of the OTOC.

To understand this result {for a general $W$}, we note that the form of Thm.~\ref{prop:otoc} is not exactly surprising. If the attack operator $W$ commutes with the circuit-evolved measurement $Z_{U} = U_\theta^\dg Z U_\theta$, then in the expectation value $y_\theta(\x)$ it will have no influence. Similarly, if the commutator is small, $({1}/{d} )\|  [Z_{U}, W ] \|_2 \ll 1$, then the adversary can only affect the outcomes of the $Z$ measurement weakly on average. That is, for an adversary to strongly affect a measurement, the strength of the commutator needs to be large. Thm.~\ref{prop:otoc} quantifies this intuition.

In the above, we note that Haar random quantum states typically lead to a concentrated expectation value $|y_\theta (\x)| \sim \frac{1}{\sqrt{d}} $~\cite{liu2020vulnerability}. This means that $\delta $ in Thm.~\ref{prop:otoc} needs to be at least larger than $\frac{1}{\sqrt{d}}$ for a successful attack, in such a typical situation. Substituting this in, we see that the dependence on $d$ drops out from the r.h.s. of the probability bound Eq.~\eqref{eq:otocLevy}. This means that even with concentration of measure effects, in this random-state setting the OTOC mediates the viability of an adversarial attack.

Now we return to the case where $W$ is a local attack. In many-body physics, the scaling of the OTOC $\braket{ [A_t, B]^2 }$ diagnoses quantum information scrambling, where $A_t$ is a time-evolving Heisenberg operator of an initially local operator $A$, and where $B$ is also a local operator. An early-time exponential growth of an OTOC indicates a fast-scrambling property of the dynamics~\cite{Sekino_Susskind_2008, Maldacena_Shenker_Stanford_2016,Swingle2018}, and the OTOC has played an important role in studies of the black hole information paradox~\cite{Shenker_Stanford_2014}, and in understanding feature of quantum chaos without a classical analogue~\cite{Swingle2016, Roberts2016, Foini2019,dowling_scrambling_2023}. Converting this to our setting of variational quantum circuits, for a scrambling circuit $U_{\theta}$ we can conclude that
\begin{equation} \label{eq:otocdef}
    \braket{ [Z_{U},W ]^2 } \sim \exp{  \left[\lambda \, \mathrm{depth}(U_\theta )\right] } \leq 1,
\end{equation}
for some $\lambda >0$. On the other hand, for a circuit that does not (quickly) scramble quantum information, this quantity scales slowly with the depth of the circuit, meaning that adversarial attacks are impossible according to Thm.~\ref{prop:otoc}.

There is evidence in the literature that trained QML architectures tend to be scrambling quantum circuits~\cite{Shen2020,We2021}, and furthermore that scrambling generally impacts trainability~\cite{garcia_quantifying_2022}. This again highlights the question of whether trained QML circuits are always vulnerable to attack? The above results depend only on the scrambling characteristic of the circuit. This concept is strictly independent of classical simulability~\cite{swingle2020}, for instance. In the following, we will investigate the robustness of a quantum classifier against attacks when $U_\theta$ instead is genuinely (quantum) chaotic, a distinct notion~\cite{dowling_scrambling_2023}.

\subsection{Universal Attacks} 
We now turn to the setting of universal adversarial attacks, i.e. a perturbation that changes the prediction of the model when applied to \textit{any} input state. 
While the existence of such perturbations is far from obvious, remarkably, they have been shown to exist in both the classical~\cite{moosavi2017universal} and quantum~\cite{gong2022universal} case. In our framework, from Eqs.~\eqref{eq:prediction} and \eqref{eq:attack} a universal attack is when for all data $\x$, 
\begin{equation}
    \mathrm{sgn}(\bra{\psi(\x)}Z_U\ket{\psi(\x)}) = -\mathrm{sgn}(\bra{\psi'(\x)} Z_U \ket{\psi'(\x)}), \label{eq:uni_spoof}
\end{equation}
where $\ket{\psi'(\x)} = \Wopt \ket{\psi(\x)}$ for some unitary $\Wopt$, which is independent of $\x$. The key point is that the allowable form of $\Wopt$ is restricted based on the encoding method, under our assumption that the adversary may only spoof the classical data $\x$. We will assume in this section that the adversary may only apply product unitaries: this is the case for angle and dense encoding. Then, we conjecture that if $Z_U$ is sufficiently \emph{complex}, then no product $\Wopt$ satisfying Eq.~\eqref{eq:uni_spoof} exists. By complex, we formally mean that it has high operator entanglement, as measured by the local-operator entanglement (LOE). LOE is a dynamical signature of chaos defined in terms of a Heisenberg operator. Through the Choi-Jamio\l kowski isomorphism, an operator $X_{U} = U_{\theta}^\dg X U_{\theta} \in \mc{B}(\mc{H})$ has a quantum state representation
\begin{equation} \label{eq:LOEdef}
    \ket{X_{U}} := ( X_{U} \otimes \id) \ket{\phi^+}, 
\end{equation}
where $\ket{\phi^+} \in \mc{H} \otimes \mc{H}$ is a maximally entangled state across a doubled Hilbert space. The LOE is then defined as the entanglement of this state, across some bipartition. This is usually measured by a R\'enyi entropy,
\begin{equation}
    S^{(\alpha)} (\omega):= \lim_{\beta \to \alpha } \frac{1}{1-\beta} \log(\tr[\omega^\beta]),
\end{equation}
for $\alpha\geq 0$, and where e.g. $S^{(1)}$ is the von Neumann entropy. 

In the remainder of this section, we will study the conjecture that operator entanglement precludes the existence of a universal adversarial attack, both analytically (Theorem~\ref{thm:universal} and Corollary~\ref{cor:clifford}) and numerically (see Fig.~\ref{fig:numerics}). 

First, to make analytic progress, we will make a simplifying assumption that the universal spoof satisfies 
\begin{equation} \label{eq:anticomm}
    Z_U  = - \Wopt ^\dg Z_U \Wopt.
\end{equation}
We call a $\Wopt$ satisfying this condition a strong universal attack. While the condition \eqref{eq:anticomm} is sufficient to satisfy Eq.~\eqref{eq:uni_spoof}, it is not necessary. This is because Eq.~\eqref{eq:anticomm} implies the stronger condition that not only the sign is flipped, but also that the magnitude of all expectation values is preserved. 
When seen in conjunction with the finite-sized numerics of a realistic model classifier (note that in the numerical setting we do not make the assumption of Eq.~\eqref{eq:anticomm}), we will argue that there is a deep relationship between LOE and the existence of a universal adversarial attack (see Fig.~\ref{fig:numerics}). A surprising facet of this setting is that even if the adversary can solve Eq.~\eqref{eq:anticomm} and determine $\Wopt$, \textit{a priori} it is not clear whether it will be implementable, given the restrictions enforced by the data encoding (here, locality), but even having dropped the assumption that the perturbation is weak.

\begin{figure*}
  \includegraphics[width=0.99\textwidth]{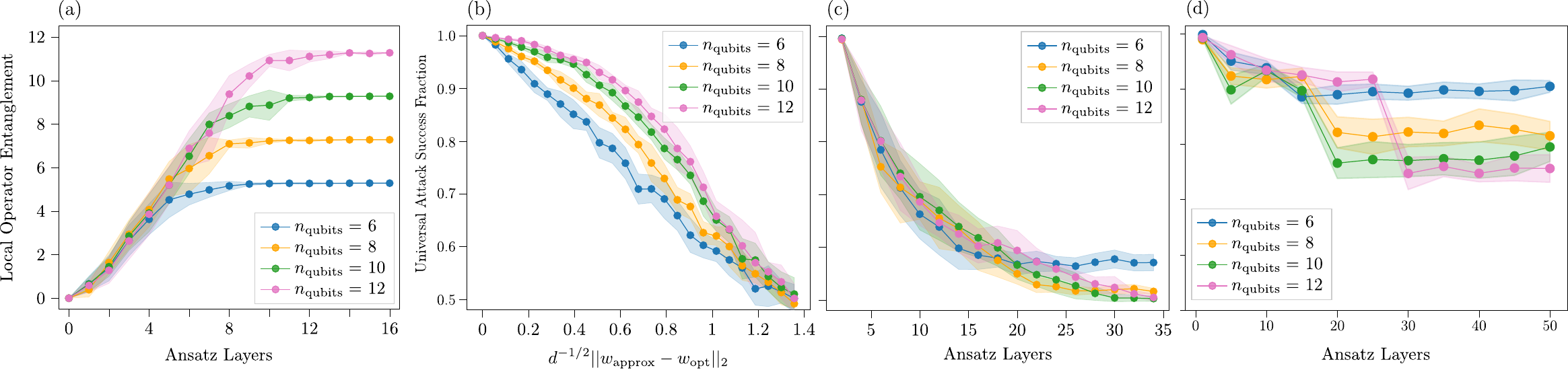}
  \caption{ \textbf{Numerical results for a common architecture.} (a) Local operator entanglement (LOE) growth in standard quantum machine learning (QML) models consisting of hardware-efficient layers of single qubit rotations and nearest neighbour CNOTs. The initial linear growth of the LOE indicates that these models are implementing chaotic quantum dynamics~\cite{dowling_scrambling_2023}.  (b) The fraction of states successfully spoofed by an approximation to a universal adversarial attack. The attack is carried out by random unitaries with various 2-norm distances from the ideal strong attack $W_\mathrm{univ}$ (satisfying Eq.~\eqref{eq:anticomm}). For each distance, we generate ten circuits, each with five attacks constructed by randomly rotating away from the ideal attack (see Eq.~\eqref{eq:approx_dist} and~\eqref{eq:rot}). The mean success fraction is plotted, with the regions within one standard deviation shaded. (c) Here, the attack is carried out by optimised local unitary operators on each qubit for random models of increasing circuit depth. For each choice of layer number, we generate $20$ circuits and train the adversary on $32,000$ training datapoints, and evaluate it on $10,000$ test datapoints. We plot the mean attack success fraction for up to $34$ layers, by which point both the LOE in the circuit and the attack success fraction have plateaued. (d) Similar to part (c), but employing a model trained to classify images of handwritten digits~\cite{mnist}. While the adversary enjoys improved performance compared to the random case, we nonetheless observe the emergence of increasing robustness with both circuit depth and qubit count. ``Universal Attack Success Fraction'' is a common vertical axis for parts (b-d).}
\label{fig:numerics}
\end{figure*}

With this in hand, we present the following result.
\begin{restatable}{thm}{universal} \label{thm:universal}
The distance $D$ between a product of local unitary channels $W=\bigotimes_i W_i$ and a strong universal adversarial attack $\Wopt$ (satisfying Eq.~\eqref{eq:anticomm}) is bounded as
    \begin{align}
      1-\ex^{- \frac{1}{2} S^{(2)}(\nu)}\leq  {D} \leq  1 - \ex^{-n S^{(2)}(\nu)}, \label{eq:universal}
      \end{align}
       where $S^{(2)}(\nu)$ is the R\'enyi 2-entropy of the reduced Choi state of a backwards circuit-evolved flip operator, $\nu := \tr_A[\Wopt \ket{\phi^+}\bra{\phi^+} \Wopt^\dagger ]$, maximised over all congruent bipartitions of the Hilbert space; $\mc{H}=\mc{H}_A\otimes \mc{H}_{\bar{A}}$. Explicitly,
       \begin{equation}
           D:= \inf_{{W}^{(i)}, 1\leq i \leq n }\frac{1}{2d} \| \mc{W}_{\mathrm{univ}} -  \bigotimes_i \mc{W}^{(i)} \|_2^2, \label{eq:distance}
       \end{equation}
       where $\mc{W} := W \otimes W^*$ denotes the (superoperator representation of the) quantum map for unitary $W$.
\end{restatable}

Here, an optimal spoof $\Wopt$ acts as a ``flip operator'' on the $k$ qubits which are measured (in the $\sigma_z$ basis) and so has only an odd number of $\sigma_x$ and $\sigma_y$ in its Pauli decomposition on these qubits, therefore flipping all of the measured expectation values $y_\theta (\x)$. A proof of Thm.~\ref{thm:universal} may be found in Section~\ref{ap:universal}.

Thm.~\ref{thm:universal} applies directly to the case of dense angle encoding \eqref{eq:encodingDense}, as the classical adversary can effectively only apply a tensor product of local unitaries. In particular, examining the two extreme cases of Eq.~\eqref{eq:universal}, we see that
 \begin{align}
        D&\approx\begin{cases}
    1,& \text{if } LOE \gg 0\\
    0,              & \text{if } LOE \approx 0.
\end{cases} \label{eq:LOE} 
\end{align}
This says that a close approximation to a universal attack is possible for a circuit with low LOE, when $D \approx 0$, in the strong sense of Eq.~\eqref{eq:anticomm}. The converse case, when the bounds of Thm.~\ref{thm:universal} are looser, is investigated numerically below. As a simple application of Thm.~\ref{thm:universal}, we consider the case where the model $U_{\theta}$ can be implemented by a Clifford circuit~\cite{Dowlin2024LOE-OSRE}. We have:
\begin{restatable}{crllr}{clifford} \label{cor:clifford}
If the variational quantum circuit $U_{\theta}$ can be implemented using only Clifford gates, then for a local data encoding map, a universal adversary exists. 
\end{restatable}

Here, Thm~\ref{thm:universal} predicts the existence of universal adversarial attacks composed of local unitaries: the flip operator $F = \sigma_x \otimes \id^{n-1}$ has zero LOE, and under conjugation with a Clifford unitary $U_\theta = C$, so does the backwards evolved $CFC^\dg$. This operation can therefore in principle be applied by an adversary through an attack on the classical data. In the case of Clifford dynamics we can also see this directly: belonging to the normaliser of the Pauli group, $C$ maps the Pauli string $F$ to the Pauli string $CFC^\dg=:W_{\mathrm{univ}}$, the required local universal attack.

Corollary~\ref{cor:clifford} applies, for example, to angle encoding, modulo the subtlety that depending on the precise details of the angle encoding, the classical adversary may not be able to apply arbitrary local unitaries. In the formulation of Eq.~\eqref{eq:encodingAngle}, for example, only linear combinations of $I$ and $X$ gates are actually realisable. Nonetheless, with high probability one of the local universal adversarial attacks will be of this form, which we discuss in detail at the end of Section~\ref{ap:universal}.

We note that the behaviour of the LOE of the model is distinct from the standard QML assumption of a $2-$design~\cite{mcclean2018barren,ragone2023unified,fontana2023adjoint}. For example, the Cliffords generate a $2-$design, but as we argue above, LOE is constant for any Clifford dynamics. In Fig.~\ref{fig:numerics}(a) we plot the scaling of the LOE in a typical hardware efficient ansatz, finding an extensive growth indicative of quantum chaos~\cite{Prosen2007,Prosen2007a, Prosen2009,Dubail_2017,Jonay2018,Alba2019,Kos2020,Alba2021,dowling_scrambling_2023}. This is consistent with previous work arguing that effective QML models are efficient scramblers~\cite{Shen2020,We2021}, a necessary condition for chaos~\cite{PhysRevLett.124.140602,dowling_scrambling_2023}. 

We stress again that a fast growing OTOC, the subject of Thm.~\ref{prop:otoc}, is not equivalent to a chaotic circuit~\cite{dowling_scrambling_2023}. This means that the combination of Thms.~\ref{prop:otoc} and \ref{thm:universal} dictate that a trained QML circuit $U_{\theta}$ is vulnerable to a universal adversarial attacks if it is sufficiently scrambling in terms of the OTOC (Thm.~\ref{prop:otoc}), yet not chaotic according to the LOE (Thm.~\ref{thm:universal}).

Finally, we also prove in Section~\ref{ap:universal} that an $\epsilon$-approximation to a universal attack (in 2-norm distance) is itself a $2\epsilon$-approximate universal adversarial attack, in the sense $\left|\bra{\psi} \left( Z_U + W^\dagger  Z_U W \right) \ket{\psi}\right| \leq 2\epsilon$ where $W$ is an $\epsilon$-approximate universal adversarial attack. In turn, this means that any state $\bm{x}$ where $|{y}_\theta(x)| > 2\epsilon$ will be misclassified after applying $W$. This implies that the range of $\epsilon$ for which an $\epsilon$-approximate universal adversarial example correctly spoofs the entire training set $S$ grows as the empirical risk $\hat{R}_S(\boldsymbol{\theta})$ of $U_\theta$ falls (see discussion below Thm.~\ref{thm:unitary}).

\subsection{Numerical Results} 
We now complement our analytic results with supporting numerical calculations, supplied in Fig.~\ref{fig:numerics}. Although Theorem~\ref{thm:universal} gives a bound between the distance from a perfect universal attack to the subset of attacks which can actually be realised by a classical adversary (i.e. consisting only of single qubit rotations) enforced by the LOE of the circuit, its operational meaning remains unclear.
For example, if $\Wopt$ is a universal attack, then so is $-\Wopt$, but $\|\Wopt-(-\Wopt)\|_2=2d\gg 0$, so being 2-norm far from a given universal attack does not guarantee that the adversary will be unsuccessful. 
We therefore now seek to connect Theorem~\ref{thm:universal} to a metric more transparently relevant in practice: the fraction of states that are misclassified following a specific attack. We investigate two scenarios: attacks a given distance in 2-norm away from a perfect attack, and optimised local attacks for QML models of varying circuit depth. 

\noindent
We begin by investigating the relationship between the distance of a given unitary from a perfect universal attack and its efficacy in practice.
For concreteness, we generate approximations $W_{\mathrm{approx}}$ to $\Wopt$ satisfying
\begin{equation}
 \| W_{\mathrm{approx}}-\Wopt \|_2 = \epsilon \sqrt{d}    \label{eq:approx_dist}
\end{equation}
for various choices of $\epsilon$
by rotating away from $W$ by a unitary generated by a random Pauli string $P$, 
i.e. 
\begin{equation}
W_{\mathrm{approx}}=e^{-it P}\Wopt e^{it P},   \label{eq:rot} 
\end{equation}
with $t$ chosen so as to satisfy Eq.~\eqref{eq:approx_dist}.
The average fraction of states successfully spoofed  by the resulting attacks is plotted in Fig.~\ref{fig:numerics}(b). 
We find a clear dependence between the success of the adversarial attack and the 2-norm distance between the corresponding unitary and $\Wopt$. For $\epsilon \sqrt{d}\approx\sqrt{2}$ (i.e. $\langle{\Wopt}|{W_{\mathrm{approx}}}\rangle_{HS}\approx 0$) the success probability drops to 1/2, with the adversarial attack faring no better than a (strong) random perturbation. 
In Section~\ref{ap:universal} we also prove that, for small $\epsilon$, we have in this setup that the fraction (with respect to the Haar measure) of states that are misclassified after an attack by $W_{\mathrm{approx}}$ is given approximately by
\begin{equation}
\Pr_{\ket{\psi}\sim \mu_{\mbu}} \left[ \mathrm{sgn}({y}_{\theta}(\ket{\psi})) \neq \mathrm{sgn}({y}_{\theta}(W_{\mathrm{approx}}\ket{\psi}))  \right]
     \approx 1 -  \frac{2\epsilon }{ \pi} 
\end{equation}
showing that a close approximation in 2-norm distance is a sufficient condition for a successful universal attack, as argued more generally in Section~\ref{ap:universal}. \\

We next consider optimised local attacks, i.e. consisting of a tensor product of parameterised single-qubit unitaries. This is a scenario that an attacker would face in practice if they had the ability to tamper with inputs before they were fed to a quantum classifier employing a local data encoding map, e.g. dense angle encoding (Eq.~\eqref{eq:encodingDense}). The attack is optimised using the \texttt{ADAM} optimiser~\cite{kingma2014adam}, and tested on  a QML model consisting of two-qubit unitary operations laid out in a brick-like fashion (see Section~\ref{ap:mps}). The simulations are performed using the matrix product state (MPS) simulator \texttt{quimb}~\cite{gray_quimb_2018} in which the bond dimension of the simulation is tracked throughout. The bond dimension of the MPS can be directly related to the entanglement entropy of the resulting state, and is maximised after $\sim n$ layers of the circuit. In Fig.~\ref{fig:numerics}(c), we consider a set of 20 random classifiers of varying circuit depth and plot the fraction of states for which the adversary can learn to induce a misclassification, finding a sharp decrease as the length of the circuits (and the entanglement present) increases. For $n>6$, the adversary fails to outperform a random perturbation long before the circuits become maximally expressible, which occurs at a depth of $\sim$$2^{n}$ layers. In the $6$ qubit case, the limiting behaviour of the adversary is to maintain a success probability greater than $0.5$. Indeed, as the system becomes small, the set of local unitaries becomes a larger portion of the total unitary group $\mathbb{U}\left(2^n\right)$, with the restrictions on the adversary relatively lessened. In Fig.~\ref{fig:numerics}(d) we consider a simple classification task where we classify the ``0'' and ``1'' handwritten digits from the MNIST dataset~\cite{mnist}. The input images were preprocessed using principal component analysis~\cite{jolliffe2016principal} to compress them to vectors of length $n_{\rm qubits}$, so that they may be angle encoded. In determining the universal attack success fraction, we take the average performance of $20$ adversaries trained using the same procedure as for the random case. Quantitatively, we find that the success probability increases compared to the random case, with a distinct transition at approximately $2N$ layers wherein the adversary goes from being moderately successful to having noticeably weakened performance. Intriguingly, we also find that in the large-depth limit, the performance of the adversary degrades with increasing system size. We do, however, note that for this simple classification task, one can train a highly accurate classifier with very few layers, obviating the need for large and deep classifiers to begin with.

\section{Discussion}
Despite considerable excitement and intense research activity, the prospect of quantum advantage in machine learning has remained questionable. Perhaps chief among the difficulties has been the discovery of barren plateaus in the training landscapes of generic variational quantum models, with a long sequence of papers~\cite{mcclean2018barren,wang2021noise,holmes2022connecting,cerezo2021cost,larocca2022diagnosing,patti2021entanglement,ragone2023unified,diaz2023showcasing,fontana2023adjoint} raising serious concerns about their trainability. 
With this phenomenon seemingly tied to the ability of the models to implement classically intractable operations, and known techniques for avoiding barren plateaus~\cite{pesah2021absence,schatzki2022theoretical,wang2023trainability,west2023provably} resulting in classically simulable circuits~\cite{cerezo2023does}, searches for alternate sources of advantage in QML are increasingly timely.
Within this context, the adversarial vulnerability of quantum models is a natural place to look, and has indeed recently attracted significant attention~\cite{lu2020quantum,liu2020vulnerability,du2021quantum,guan2021robustness,weber2021optimal,liao2021robust,kehoe2021defence,west2023towards,wu2023radio,west2023benchmarking,west2023drastic,berberich2023training,khatun2024quantum,winderl2024constructing,ren2022experimental}.

While it is not \textit{a priori} clear why one would expect quantum dynamics to be suited to implementing machine learning algorithms more efficiently than classical methods, guarantees against spoofing are far more in line with capabilities that naturally arise in, for example, quantum communication~\cite{bennett2014quantum}. Similarly, our results do not depend on the details of the model being secret, and are guaranteed purely by its quantum mechanical nature. 
In the universal adversarial attack case, for example, even if the classical adversary knows exactly what $\Wopt$ is, they are unable to apply any attack close to it when the corresponding LOE is high enough and a local data encoding strategy is employed.

The theoretical framework we have developed in this work attempts to examine the vulnerability of quantum classifiers in the context of a practically relevant threat model: that of an adversary who can manipulate classical data before it is sent to the quantum computer for encoding. Given the various existence proofs of adversarial examples for quantum models that can classify arbitrary states~\cite{liu2020vulnerability} or states smoothly generated from a Gaussian latent space~\cite{liao2021robust}, it is an interesting open question to characterise exactly which classes of states, and exactly which variational circuits, lead to quantum models with provable robustness guarantees. Moreover, it is a key question to understand to what extent ``active'' robustness methods can improve upon our ``passive'' results, whereby, e.g. encoding can be tailored to be resistant to kinds of adversarial attacks~\cite{Gong2024}. While our focus has been on the impossibility of implementing adversarial perturbations given only access to the classical data, showing, for example, that encoding schemes that do not produce entanglement yield provable guarantees against universal adversarial attacks,
also of interest is the computational cost of finding such attacks in the first place. Future searches for robustness guarantees could seek to connect the difficulty of spoofing with the difficulty of simulating the classifier itself, or that of understanding the data that is being classified.

It is also interesting, in the current era of noisy intermediate-scale (NISQ) quantum computers, to investigate the validity of our results in the presence of uncontrolled external noise. In this case, one would have to instead take $U_\theta \mapsto \mc{L}_\theta$, a completely positive trace preserving (CPTP) map, which generally describes open quantum system evolution~\cite{Wilde_2013}. Our results in the first two columns of Table~\ref{tab:results} readily hold also in this paradigm. In particular, the bound of Thm.~\ref{thm:unitary} immediately holds also for CPTP map classifiers $\mc{L}_\theta$, as we show in Section~\ref{ap:weak_proof}. In fact, as trace preserving maps can not increase the distance between quantum states~\cite{nielsen2010quantum}, with the distance remaining constant if and only if the dynamics are unitary, Thm.~\ref{thm:unitary} is in fact strengthened by the presence of noise-induced non-unitarity. A similar point was made in Ref.~\cite{du2021quantum} in the context of depolarisation noise, the strength of which was linked to differential privacy, a measure of the insensitivity of a map to changes in its input~\cite{zhou2017differential}. The key difficulty with such guarantees, strengthened though they are by increasing noise, is separating the (non-perturbed) predictions from zero in the first place. Further, in terms of Thm.~\ref{prop:otoc}, the only quantity in Eq.~\eqref{eq:otocLevy} depending on the circuit is the OTOC, which also tends to grow fast in noisy, scrambling systems~\cite{Zhang2019,PhysRevLett.131.160402}. For our results in the final column of Table~\ref{tab:results}, it is less clear as the LOE has not been well studied in the open systems setting. Finally, we note that in Ref.~\cite{mele2024noiseinduced} it is shown that noisy circuits are generally only be as useful as shallow circuits for computing expectation values of Pauli operators, which is what QML is largely concerned with. This suggests that noise-tolerance will be a necessary ingredient of QML advantage, before one needs to consider adversarial robustness.  

On this note, it remains to be seen to what extent the chaos-based quantum guarantees of adversarial robustness are consistent with the trainability of the safeguarded model in the first place. Highly entangling operators can suffer from entanglement-induced barren plateaus~\cite{marrero2021entanglement} (although this is not guaranteed~\cite{ragone2023unified}), as does the related problem of trying to learn a given scrambling operator~\cite{Holmes2021scrambling}. As the theory of the trainability of variational QML models continues to advance~\cite{ragone2023unified, fontana2023adjoint, diaz2023showcasing}, the search for models that maximise trainability and robustness while minimising classical simulability remains an important research direction.

\newpage
\onecolumngrid

\section{Methods}
Here we provide the proofs of all theorems presented in the main text, as well as further details of our techniques.

\subsection{Weak Targeted Attack} \label{ap:weak_proof}
\unitary*
\begin{proof}
To induce a misclassification, an adversary needs to perturb the result of the final $Z$ measurement by at least $|{y}_{\theta}(\x)|$. For shorthand we write the original encoded state as $\drho(\x) := \ket{\psi(\x)} \bra{\psi(\x)}$, the attacked state as $\drho(\x') := W \ket{\psi}\bra{\psi} W^\dg$.
Then, $\Pi_{\vec{i}} := \ket{i_1}\bra{i_1} \otimes \ket{i_2}\bra{i_2} \otimes \dots$ is a projection onto the computational basis, and $Z \equiv \sigma_z \otimes \sigma_z \otimes \dots$ is a $Z$-basis measurement on $k$ qubits. For a successful attack, we require
\begin{align}
     |{y}_{\theta}(\x)| &\leq |\Delta Z|  \\
     &= |\braket{Z}_{{U}(\drho(\x) - \drho(\x')){U}^\dg}|  \\
     &=| \tr[{Z} ({U}(\drho(\x) - \drho(\x)) {U}^\dg )]|\\
     &= \big|\sum_{\vec{i}} (-1)^{i_1 +i_2+i_3+\dots}\tr[\Pi_{\vec{i}} ({U}(\drho(\x) - \drho(\x')){U}^\dg )]\big|\\
    &\leq \sum_{\vec{i}} \left|\tr[\Pi_{\vec{i}} ({U}(\drho(\x) - \drho(\x')){U}^\dg )]\right|  \label{eq:linear}\\
    &\leq \max_{\{ P_i \} } \sum_{{i}} \left|\tr[P_{{i}} ({U}(\drho(\x) - \drho(\x')){U}^\dg )]\right|  )  \\
    &=  \| U (\drho(\x) - \drho(\x')) {U}^\dg\|_1 \label{eq:A7}\\
    &= { \| \drho(\x) - \drho(\x') \|_1}  \label{eq:A8}
\end{align}
where $\{ P_i\}$ is an arbitrary POVM, and the inequality fourth line comes from the fact that diagonal projections in the computational basis form a POVM, with $\sum_{\vec{i}}\Pi_{\vec{i}} = \id$. Here we have used the operational definition of the trace distance, and that the $1-$norm distance is unitarily invariant. 

Finally, this proof can be extended to to the case of general (non-unitary) evolution. Such evolution, which can e.g. describe noisy circuits, are described by CPTP maps characterised by Kraus operators~\cite{Wilde_2013}, where for some density matrix input $\rho$, 
\begin{equation}
    \rho \underset{\text{CPTP}}{\mapsto} \mc{L}(\rho) := \sum_j K_j \rho K_j^\dagger ,
\end{equation}
with $\sum_j K_j K_j^\dagger = \id$. Then the above proof can be amended by replacing the unitary evolution $U (\cdot) U^\dagger$ with $\sum_j K_j^\dagger (\cdot)) K_j$. Then going from \eqref{eq:A7} to \eqref{eq:A8} we can use the fact that trace norm distance is contractive under CPTP maps~\cite{nielsen2010quantum}, to arrive at the same result
\begin{align}
    |{y}_{\theta}(\x)| \leq \dots \leq \| \sum_j K_j (\drho (\x) - \drho (\x')) K_j^\dagger \|_1 \leq { \| \drho (\x) - \drho (\x') \|_1}.
\end{align}
\end{proof}

The operational significance of this result depends on how changes in the classical vector to be encoded translate to changes in the corresponding quantum state, and will depend on the specific encoding technique employed. We now undertake this analysis for two common encoding schemes and two different attack methods: (i) angle, (ii) dense angle, and (ii) amplitude encoding, under an adversarial attack $\x\mapsto\x+\epsilon\boldsymbol{w}$, with $\max_i |w_i| \sim 1$ and $|\epsilon|\ll 1$.
We begin by noting that for $\psi = \ket{\psi}\bra{\psi}$ and $\phi = \ket{\phi}\bra{\phi}$,
\begin{equation}
    \frac{1}{2}\|\drho - \phi\|_1^2 = 1 - |\braket{\psi|\phi}|^2.     \label{eq:one-norm}
\end{equation}
(i) First for angle encoding:
\begin{align}
    \left\lvert\braket{\psi_{\mathrm{angle}}(\x)|\psi_{\mathrm{angle}}(\x+\epsilon\boldsymbol{w})}\right\rvert^2&=\big\lvert\braket{\psi_{\mathrm{angle}}(\x)| \bigotimes_i w_i | \psi_{\mathrm{angle}}(\x)}\big\rvert^2 \label{eq:angleWeak1}\\ 
    &= \left\lvert\prod_{j=1}^N \bra{0}_j   \mathrm{e}^{i x_j \sigma_x} \mathrm{e}^{-i (x_j+\epsilon w_j) \sigma_x} \ket{0}_j\right\rvert^2 \\
    &=\prod_{j=1}^N \left\lvert\bra{0}_j  \mathrm{e}^{-i \epsilon w_j \sigma_x} \ket{0}_j\right\rvert^2 \label{eq:angleWeakTaylor1} \\
    &\approx \prod_{j=1}^N \left(1- \frac{\epsilon^2 w_j^2}{2} \right)^2 \label{eq:angleWeakTaylor2}\\
    &\approx 1 - N \epsilon^2 + \mc{O}\left(N^2\epsilon^4\right) \label{eq:angleWeak2}
\end{align}
where we have taken $\epsilon \ll 1$, and that $w_j \approx w_k \approx 1$ for any $1\leq j,k \leq N $. In going from Eq.~\eqref{eq:angleWeakTaylor1} to Eq.~\eqref{eq:angleWeakTaylor2} we have used that the linear term in the Taylor series expansion vanishes as $\bra{0}   \sigma_x \ket{0} =0$. Therefore, to lowest order in $\epsilon$, 
\begin{equation}
    \|\drho(\x) - \drho(\x') \|_1 =2 \sqrt{N} \epsilon + \mc{O}(\epsilon^2).
\end{equation}

(ii) Next we turn to dense angle encoding. Temporarily adopting the notation
$\ket{\theta,\phi}$ for a single qubit state with Bloch sphere angles ${\theta,\phi}$ and recalling the relation~\cite{nielsen2010quantum}
\[\left\lvert\braket{\theta,\phi|\alpha,\beta}\right\rvert^2 = 1 - \frac{1}{4}  \ \left|\left|\ \begin{pmatrix} \sin(\theta)\cos(\phi) \\  \sin(\theta)\sin(\phi) \\ \cos(\theta) \end{pmatrix} - \begin{pmatrix} \sin(\alpha)\cos(\beta) \\  \sin(\alpha)\sin(\beta) \\ \cos(\alpha) \end{pmatrix}  \ \right|\right|_2^2 \]
between the  fidelity of two single qubit states and the Euclidean distance between their Bloch sphere vectors, we have

\begin{align}
    \left\lvert\braket{\psi_{\mathrm{dense}}(\x)|\psi_{\mathrm{dense}}(\x+\epsilon\boldsymbol{w})}\right\rvert^2&=\big\lvert\braket{\psi_{\mathrm{dense}}(\x)| \bigotimes_i w_i | \psi_{\mathrm{dense}}(\x)}\big\rvert^2 \label{eq:denseWeak1}\\        
    &= \prod_{j=1}^{N/2} \left\lvert \bra{0}_j \mathrm{e}^{i x_{2j-1} \sigma_y} \mathrm{e}^{i x_{2j} \sigma_z}\mathrm{e}^{-i (x_{2j}+\epsilon w_{2j}) \sigma_z}   \mathrm{e}^{-i (x_{2j-1}+\epsilon w_{2j-1}) \sigma_y} \ket{0}_j\right\rvert^2 \\
    &= \prod_{j=1}^{N/2} \left\lvert \braket{x_{2j-1},x_{2j}|x_{2j-1}+\epsilon w_{2j-1},x_{2j}+\epsilon w_{2j}}\right\rvert^2 \\
    &= \prod_{j=1}^{N/2} \bigg( 1- \frac{1}{4}  
    \left|\left|\ 
    \begin{pmatrix}
        \sin(x_{2j-1}+\epsilon w_{2j-1})\cos(x_{2j}+\epsilon w_{2j}) \\ 
        \sin(x_{2j-1}+\epsilon w_{2j-1})\sin(x_{2j}+\epsilon w_{2j}) \\
        \cos(x_{2j-1}+\epsilon w_{2j-1})
    \end{pmatrix}
    -
    \begin{pmatrix}
        \sin(x_{2j-1})\cos(x_{2j}) \\ \sin(x_{2j-1})\sin(x_{2j}) \\ \cos(x_{2j-1})
    \end{pmatrix}
    \ \right|\right|_2^2  \bigg) \\
    &\approx \prod_{j=1}^{N/2} \bigg( 1- \frac{1}{4} \left|\left|\  \epsilon 
    \begin{pmatrix}
        w_{2j-1}\cos(x_{2j-1})\cos(x_{2j}) -  w_{2j}\sin (x_{2j-1})\sin (x_{2j})\\
         w_{2j-1}\cos(x_{2j-1})\sin (x_{2j}) +  w_{2j}\sin (x_{2j-1})\cos (x_{2j} )\\
        - w_{2j-1}\sin(x_{2j-1}) 
    \end{pmatrix}
    \ \right|\right|_2^2 \bigg) \\
    &= \prod_{j=1}^{N/2} \left( 1- \frac{\epsilon^2}{4} \left( w_{2j-1}^2 + \sin^2(x_{2j-1}) w_{2j}^2 \right)\right)\\
    &\sim  \prod_{j=1}^{N/2} \left(1-\epsilon ^2 \right)\\
    &\sim 1-N\epsilon ^2 + \mc{O}\left(N^2\epsilon^4\right)\\
\end{align}
again using that $w_j \sim 1 \ \forall j\in \{1,\ldots, N\}$.
So, as in the case of angle encoding, we find
\begin{equation}
    \|\drho(\x) - \drho(\x') \|_1 =  \sqrt{N} \epsilon + \mc{O}(\epsilon^2).
\end{equation}

(iii) Now for amplitude encoding,
\begin{align}
    |\braket{\psi_{\mathrm{amp}}(\x)|\psi_{\mathrm{amp}}(\x+\epsilon\boldsymbol{w})}|^2 &=\left| \frac{\sum_j \bra{j} x_j \sum_k (x_k+\epsilon w_k ) \ket{k}}{ \sqrt{\sum_n |x_n|^2} \sqrt{\sum_m |x_m + \epsilon w_m|^2 }} \right|^2 \label{eq:amp1}\\
    &= \frac{\left|\sum_j x_j^2 +\epsilon x_j w_j \right|^2}{|\x|^2 {\sum_m x_m^2 + 2 \epsilon w_m x_m+  \epsilon^2 w_m^2 }} \\
    &= \frac{\left||\x|^2 +\epsilon  \braket{\x,\boldsymbol{w}} \right|^2}{|\x|^2 ({|\x|^2+ 2 \epsilon  \braket{\x,\boldsymbol{w}}+  \epsilon^2 |\boldsymbol{w}|^2 } )} \\
    &=\left(|\x|^2 + 2  \epsilon \braket{\x,\boldsymbol{w}} +  \epsilon^2 \frac{|\braket{\x,\boldsymbol{w}}|^2}{| \x|^2} \right)  \left(\frac{1}{|\x|^2} - 2 \epsilon \frac{\braket{\x,\boldsymbol{w}} }{|\x|^4} + \epsilon^2 \frac{4 \braket{\x,\boldsymbol{w}}^2 - |\x|^2 |\boldsymbol{w}|^2}{|\x|^6} + \mathcal{O}(\epsilon^3) \right) \\
    &=1 + \epsilon \left( 2 \frac{ \braket{\x,\boldsymbol{w}} }{ |\x|^2} - 2  \frac{ \braket{\x,\boldsymbol{w}} }{ |\x|^2}  \right) + \epsilon^2 \left( -4 \frac{\braket{\x,\boldsymbol{w}}^2}{|\x|^4} + \frac{\braket{\x,\boldsymbol{w}}^2}{|\x|^4 } + \frac{4 \braket{\x,\boldsymbol{w}}^2 - |\x|^2 |\boldsymbol{w}|^2 }{|\x|^4}   \right) + \mathcal{O}(\epsilon^3) \\
    &=1 + \epsilon^2 \left( \frac{\braket{\x,\boldsymbol{w}}^2}{|\x|^4}  -\frac{|\boldsymbol{w}|^2 }{|\x|^2}   \right) + \mathcal{O}(\epsilon^3)
\end{align}
where we have employed the second order Taylor series expansion of $1/(a+b\epsilon+c\epsilon^2)$. Now, if we choose $\epsilon$ such that $|\boldsymbol{w}|^2  = |\x|^2$, then for small $\epsilon$
\begin{equation}
    |\braket{\psi_{\mathrm{amp}}(\x)|\psi_{\mathrm{amp}}(\x+\epsilon\boldsymbol{w})}|^2 = 1 - \epsilon^2 \left( 1-\frac{\braket{\x,\boldsymbol{w}}^2}{|\x|^4}   \right) + \mathcal{O}(\epsilon^3).
\end{equation}
Then, again using the pure state identity Eq.~\eqref{eq:one-norm}, to first order in $\epsilon$
\begin{equation}
    \|\drho(\x) - \drho(\x') \|_1 = \epsilon \sqrt{2 \left(1- \frac{\braket{\x,\boldsymbol{w}}^2}{|\x|^4}  \right)}  + \mc{O}(\epsilon^2).
\end{equation}
We see that this time the resulting expression does not scale with $N$, implying that a weak perturbation of each element of the classical vector leads to a weakly perturbed encoded state, irrespective of the dimension of data.

\subsection{(Strong) Local Attacks} \label{ap:local}
Here, we give some extra background and details on quantum information scrambling, and detail the proofs of the analytic result pertaining to local adversarial attacks.  

First, we will argue that changing a small number of bits $\ell \ll n$ of the classical data string $\x$ leads to a weakly perturbed state after amplitude encoding, with $|\x|_0 = n$, the length of the data bit-string. 
This will mean that Thm.~\ref{thm:unitary} can be applied directly to this case. Recall the effect of a weak perturbation in Eq.~\eqref{eq:amp1},
\begin{align}
    |\braket{\psi_{\mathrm{amp}}(\x)|\psi_{\mathrm{amp}}(\x+\epsilon\boldsymbol{w})}|^2 &=\left| \frac{\sum_j \bra{j} x_j \sum_k (x_k+\epsilon w_k ) \ket{k}}{ \sqrt{\sum_n |x_n|^2} \sqrt{\sum_m |x_m + \epsilon w_m|^2 }} \right|^2 
\end{align}
A local attack corresponds to substituting in the above 
 \begin{align}
        \epsilon w_i&\mapsto\begin{cases}
    w_i,&  i \in [a,a+\ell] \\
    0,              & i \notin  [a,a+\ell].
\end{cases} \label{eq:localCondition} 
\end{align}
Here, we have assumed that the encoding maps all the attacked pixels to adjacent bits, which we are free to choose. Then, we have that 
\begin{align}
    |\braket{\psi_{\mathrm{amp}}(\x)|\psi_{\mathrm{amp}}(\x')}|^2 &=\left| \frac{\sum_j \bra{j} x_j \sum_k (x_k+\epsilon w_k ) \ket{k}}{ \sqrt{\sum_n |x_n|^2} \sqrt{\sum_m |x_m + \epsilon w_m|^2 }} \right|^2 \\
    &= \left| \frac{1 + \frac{\langle \mathbf{w} , \mathbf{x}  \rangle }{| \x |^2 }}{\sqrt{1 + \frac{|\mathbf{w}|^2}{|\x|^2} + \frac{2 \langle \mathbf{w} , \mathbf{x}  \rangle }{| \x |^2 }} } \right|^2.
\end{align}
Now, assuming that $|x_i|,|w_i| \leq 1$, from counting arguments both
\begin{equation}
    \frac{\langle \mathbf{w} , \mathbf{x}  \rangle }{| \x |^2 } = 
    \frac{1 }{| \x |^2 }\sum_{i=a}^{a+\ell} x_i w_i  
    \sim \frac{\ell}{n} :=\epsilon 
\end{equation}
and similarly
\begin{equation}
    \frac{|\mathbf{w}|^2}{|\x|^2} \approx \epsilon.
\end{equation}
Following a similar argument to Eq.~\eqref{eq:angleWeak1}-\eqref{eq:angleWeak2}, then 
\begin{equation}
     \| \ket{\psi_{\mathrm{amp}}(\x)} \bra{\psi_{\mathrm{amp}}(\x)} -  \ket{\psi_{\mathrm{amp}}(\x')} \bra{\psi_{\mathrm{amp}}(\x')} \|_1 = \epsilon + \mc{O}(\epsilon^2).
\end{equation}
A similar argument does not apply to angle encoding. 

Now, we will prove a relation between quantum information scrambling and (local) Adversarial Attacks.

\localOTOC*

\begin{proof}
Now, let us assume that an adversary may change the initial quantum state by some arbitrary operator $w$. We take an encoding-agnostic approach to this result, and stress that a similar result should hold upon specifying a particular scheme. 
Recall the original and spoofed expectation values 
\begin{align}
    &y_\theta(\x) =\braket{\psi(\x) |   Z_{U}  |\psi(\x)  }, \text{ and,}\\
    &y_\theta (\x') =\braket{\psi(\x) | W^\dg Z_{U} W |\psi(\x)  },
\end{align}
where both $W$ and $Z$ (not circuit-evolved) are (relatively) local operators, and both $\psi(\x)$ and $W=W(\x)$ may depend on the classical state $\x$ (i.e. the original data). Here, $Z_{U} = U_\theta^\dg Z U_\theta $. Now, we want to know whether the adversary can spoof using only local but possibly strong `attacks' (operators) $W$. 

We consider values of $m'$, with a sampling of the initial state $\ket{\psi} = V \ket{0}$ over a unitary 2-design $\mce$. If we want to use the full concentration of measure results -- as opposed to using Chebyshev's inequality -- it would need to instead be over the full Haar measure. We can solve both for the average and variance of $y_{\theta}(\x) - y_{\theta}(\x')$ over this ensemble. Then, if we can bound the variance as $\sigma^2 \leq X$ (or just compute the variance exactly), then from Chebyshev we know that for $\delta >0$
\begin{equation}
    \Pr \left\{ |y_{\theta}(\x) - y_{\theta}(\x') - \mu | \geq \delta' \sqrt{X} \right\} \leq \frac{1}{(\delta')^2}.
\end{equation}
Then the average is 
\begin{align}
    \mathbb{E}_{V \sim \mce} ( y_{\theta}(\x) - y_{\theta}(\x') ) &= \mathbb{E} ( \braket{\psi(\x) | V^\dg    Z_{U}  V |\psi(\x)  } - \braket{\psi(\x) | V^\dg  W^\dg  Z_{U}  W V |\psi(\x)  } )\\
    &=\tr[Z_{U} ] - \tr[W^\dg  Z_{U}  W] =0,
\end{align}
given the unitarity of $U_\theta, W$, and the traceless property of $Z$. For clarity we have dropped the notation of the measure dependence. Now, the variance is
\begin{align}
    \sigma^2 &= \mathbb{E}((y_{\theta}(\x) - y_{\theta}(\x'))^2) - \mathbb{E}((y_{\theta}(\x) - y_{\theta}(\x')))^2  \\
    &= \mathbb{E}(y_{\theta}(\x)^2 -2 \mathrm{Re}(y_{\theta}(\x)  y_{\theta}(\x')) + y_{\theta}(\x')^2). \label{eq:variance}
\end{align}
Here we can use the explicit expression for a $2-$fold average over the Haar ensemble, derivable from Weingarten calculus~\cite{mele2023introduction}. For some tensor $X \in \mc{H}\otimes \mc{H}$ this is~\cite{Roberts2017-en}
\begin{align}
    &\Phi^{(2)}_\mathrm{Haar}(X):= \int dU U \otimes U (X) U^\dg \otimes U^\dg \label{eq:2fold}\\
    &\quad= \frac{1}{d^2-1} \left(\id \tr[X] + \mathbb{S} \tr[\mathbb{S} X] -\frac{1}{d}\mathbb{S} \tr[X] -\frac{1}{d} \id \tr[\mathbb{S} X] \right), \nn
\end{align}
Note also that for any $X$, by definition the $2-$fold Haar average is equal to the $2-$fold average over the unitary $2-$design $\mce$, that is $\Phi^{(2)}_\mathrm{Haar}(X) = \Phi^{(2)}_\mce(X)$. 

Then handling the terms of Eq.~\eqref{eq:variance} one at a time,
\begin{align}
    \mathbb{E}(y_{\theta}(\x)^2) &= \frac{1}{d^2-1} \left( \tr[Z_{U} ]^2 \braket{\psi(\x)|\psi(\x)}^2 + \tr[Z_{U}^2 ] \braket{\psi(\x)|\psi(\x)}^2 -\frac{1}{d}\tr[Z_{U}^2 ] \braket{\psi(\x)|\psi(\x)}^2 -\frac{1}{d} \tr[Z_{U} ]^2  \braket{\psi(\x)|\psi(\x)}^2  \right) \\
    &= \frac{1}{d^2-1} \left( \tr[\id] -\frac{1}{d}\tr[\id ]  \right) = \frac{1}{d+1}.
\end{align}
Then it is easy to check that $\mathbb{E}(y_{\theta}(\x)^2) = \mathbb{E}(y_{\theta}(\x')^2) $ by repeating the above calculation but with $Z_{U} \to W^\dg Z_{U} W$. The non-trivial component is then 
\begin{align}
    \mathbb{E}( \mathrm{Re}(y_{\theta}(\x)  y_{\theta}(\x'))) &=  \frac{1}{d^2-1} \bigg( \tr[Z_{U} ] \tr[W^\dg Z_{U} W ]  + \tr[Z_{U} W^\dg Z_{U} W] -\frac{1}{d} \tr[Z_{U} W^\dg Z_{U} W]\nn \\ 
    &\hspace{47.5mm}-\frac{1}{d} \tr[Z_{U} ] \tr[W^\dg Z_{U} W ]  \bigg) \braket{\psi(\x)|\psi(\x)}^2 \\
    &=\frac{1}{d (d+1)}\tr[Z_{U} W^\dg Z_{U} W].
\end{align}
Remarkably, $\frac{1}{d}\tr[Z_{U} W^\dg Z_{U} W]$ (with $W, Z$ being local) is exactly the OTOC: a measure of quantum information scrambling. 

Subbing this into Chebyshev's inequality, and choosing $\delta' = \delta \sqrt{\frac{2}{d+1} (1 - \frac{1}{d}\mathrm{Re}(\tr[Z_{U} W^\dg Z_{U} W]))} >0$
\begin{align}
    &\Pr \left\{ |y_{\theta}(\x) - y_{\theta}(\x') | \geq \delta' \sqrt{\frac{2}{d+1} (1 - \frac{1}{d}\mathrm{Re}(\tr[Z_{U} W^\dg Z_{U} W]))}\right\} \leq \frac{1}{(\delta')^2} \\
    \iff& \Pr \left\{ |y_{\theta}(\x) - y_{\theta}(\x') | \geq \delta\right\} \leq \frac{2 (1 - \frac{1}{d}\mathrm{Re}(\tr[Z_{U} W^\dg Z_{U} W]))}{(d+1)\delta^2} = \frac{\braket{ [Z_{U},W ]^2 }}{(d+1)\delta^2} \sim \frac{ \exp{  \left[\lambda \, \mathrm{depth}(U_\theta )\right]}}{(d+1)\delta^2} 
\end{align}
Here, we have chosen $\delta' $ such that this bound corresponds to the necessary change $|y_{\theta}(\x) - y_{\theta}(\x') |$ required to spoof the outcome is $\delta$. Note that the exponential in the final ``$\sim$'' is for early time behaviour of scrambling systems, and that $0\leq \braket{ [Z_{U},W ]^2 }\leq 1$. 
\end{proof}

One could extend the above result to specific encoding schemes. For example, one could replace the Haar averaging above with an averaging over the local angles of Pauli-$Y$ rotations in angle encoding \eqref{eq:encodingAngle}. This results in a more complex expression, with terms proportional to different OTOCs and related quantities.

\subsection{(Strong) Universal Attack}\label{ap:universal}
Here, we give further details on the robustness of QML circuits against universal attacks, from the chaoticity of the circuit.  

First, we provide a proof of our main technical result in this setting.
\universal*
\begin{proof}
We assume that an adversary only has access to modifying the classical input $ \x $. Note that we do not assume anything about the weakness of the attack, in contrast to Prop.~\ref{thm:unitary}, but rather that a single attack should flip the prediction of all inputs $\x$. Mathematically, we replace the classical input data $\x$ with
\begin{equation}
     \x \underset{\text{attack}}{\rightarrow}  \x^\prime =:  \x+ \w
\end{equation}
where $\boldsymbol{w}$ is independent of $\x$.
Rewriting the full algorithm under the effects of the attack, 
\begin{equation}
    \x^\prime \underset{\text{encode}}{\rightarrow} \ket{\psi(\x^\prime)} = W \ket{\psi(\x)}  = \underset{\text{\phantom{..}QVC\phantom{..}}}{\rightarrow} U_{\theta} \ket{\psi(\x^\prime)} \underset{\text{measure}}{\rightarrow} \tr[(Z \otimes \id) U W \ket{\psi}\bra{\psi}W^\dg U^\dg] =: y_\theta (\x')  \label{eq:QML_sequence_spoof}.
\end{equation}
The specific form of the induced unitary $W$ will depend on the form of data encoding employed; for example, for dense encoding (Eq.~\eqref{eq:encodingDense}) which for convenience we restate here: 
\begin{equation}
    W \ket{\psi(\x)} = \bigotimes_{j=1}^{n/2} \mathrm{e}^{-i (x_{2j}+ w_{2j}) \sigma_z}\mathrm{e}^{-i (x_{2j-1}+ w_{2j-1}) \sigma_y} \ket{0}_j ,
\end{equation}
$W$ is of the form of a tensor product of single qubit unitaries. More generally, whenever
the adversary only has access to the classical data and a local encoding is employed, $W=W_1 \otimes W_2 \otimes \dots \otimes W_N$. This motivates the investigation of the quantity Eq.~\eqref{eq:distance} as the distance between a unitary that implements a universal adversarial attack, and the class of unitaries which the adversary is actually able to implement. 

In the universal adversarial attack scenario, the adversary wishes to simultaneously change the prediction of \textit{all} input states. With the predictions given by Eq.~\eqref{eq:prediction}, this implies
\begin{align}
    &y_{\theta}(\x+\boldsymbol{w}) {=} -y_{\theta}(\x) \hspace{72.5mm}\forall\x\\
    \iff& \tr[Z U W \ket{\psi(\x)}\bra{\psi(\x)}W^\dg U^\dg] = - \tr[Z  U \ket{\psi(\x)}\bra{\psi(\x)}U^\dg] \hspace{10mm}\forall\x\\
    \iff& 0= \bra{\psi(\x)} U^\dg  [ Z + \WU{U^\dg}^\dg Z \WU{U^\dg} ] U \ket{\psi(\x) } \hspace{41.5mm}\forall\x\label{eq:spoofed}
\end{align}
where $\WU{U^\dg} := U W U^\dg $ and $Z := \sigma_z^{\otimes k} \otimes \id^{n-k}$ is a Pauli-$Z$ measurement on the first $k$ qubits. For a universal adversary (with respect to the encoding of Eq.~\eqref{eq:encodingDense}), this equation must be satisfied for any initial product state, and thus by linearity for any state. In light of this, we conclude that 
\begin{equation}
    Z + \WU{U^\dg}^\dg Z \WU{U^\dg} = {0}. \label{eq:op_spoof}
\end{equation}
The general $\WU{U^\dg}$ that satisfies Eq.~\eqref{eq:op_spoof} is
\begin{equation}
    \WU{U^\dg} = \sum_{ij} c_{ij} F^{(k)}_i \otimes P_j, \label{eq:general_flip}
\end{equation}
where $F^{(k)}_i $ is a flip operator on the $k$ qubits which $Z$ measures, i.e. in the Pauli basis, $F^{(k)}_i $ has an odd number of $\sigma_x$ and $\sigma_y$ local basis elements, ensuring that it anti-commutes with $\sigma_z^{\otimes k} $. The index $i$ iterates over the Pauli strings satisfying this, and $j$ over all Pauli strings of length $n-k$. Now, what does this tell us about the spoofing operator $W$? We know that $W$ is restricted to be a product operator, due to our assumption on the dense angle encoding method \eqref{eq:encodingDense}. We will investigate how close the adversary can get to a perfect universal spoof, through the Hilbert-Schmidt distance
\begin{align}
   D :=&  \inf_{W_1,\dots, W_n} \frac{1}{2d^2} \| \Wopt \otimes \Wopt^* - (W_1 \otimes  \dots \otimes W_n)\otimes (W_1^* \otimes  \dots \otimes W_n^*)  \|_2^2, \label{eq:distance_spoof}
\end{align}
where the normalization of $1/2d^2$ is introduced such that $0 \leq D \leq 1$. $D$ measures how close the adversary can get to a universal spoofing. From Eq.~\eqref{eq:distance_spoof}, $\Wopt$ is the (backwards-)time-evolved flip operator $\Wopt:= U (F_k \otimes \id) U^\dg$. We have that 
\begin{align}
     D &=  \inf_{W_1,\dots, W_n} \frac{1}{2d^2} (\tr[\Wopt^2]^2 + \tr[(W_1 \otimes  \dots \otimes W_n)^2]^2 - 2 \tr[\Wopt ( W_1 \otimes  \dots \otimes W_n)^\dg ]^2  ) \label{eq:proof1b} \\
    &= \inf_{W_1,\dots, W_n} \frac{1}{2d^2} (2 d^2 - 2d^2 \langle{\Wopt } | { W_1 \otimes  \dots \otimes W_n} \rangle^2  ) \label{eq:proof1c} \\ 
    &= 1- |\lambda_{\infty}^{(1)} \lambda_{\infty}^{(2)} \cdots \lambda_{\infty}^{(n)}|^2  = 1 - \ex^{- S^{(\infty)}(\mc{H}_{1}:\mc{H}_{2:n})} \ex^{-S^{(\infty)}(\mc{H}_{1:2}:\mc{H}_{3:n})} \dots \ex^{-S^{(\infty)}(\mc{H}_{1:n-1}:\mc{H}_{n}) }\label{eq:proof1d}
\end{align} 
In the third line we have used that $\tr[\Wopt ( W_1 \otimes  \dots \otimes W_n)^\dg ] = d \bra{\phi^+} \Wopt ( W_1 \otimes  \dots \otimes W_n)^\dg \ket{\phi^+}$, and defined the Choi states $\ket{ \Wopt} := \Wopt \ket{\phi^+} $ for the normalised bell state $\ket{\phi^+}:= 1/\sqrt{d} \sum_j \ket{jj}$. Then $\ket{W_1 \otimes  \dots \otimes W_n}$ is a product state as we assume $W_i$ to be unitary (the effective ability of an adversary under dense angle encoding \eqref{eq:encodingDense}). In the penultimate line we have used that the largest fidelity of a bipartite state with any product state corresponds to the largest singular value $|\lambda_{\infty}|^2$ from the Schmidt decomposition across the bipartition. Then, generalised to the closest multipartite state, the largest fidelity corresponds the product of the largest singular values $|\lambda_\infty^{(i)}|^2$, across all congruent bipartitions of the first $i$ qubits: the next $n-i$ qubits. $S^{(\infty)}(\mc{H}_{i}:\mc{H}_{i+1:n}) := -\log(|\lambda_\infty^{(i)}|^2)$ is the min-entropy of the reduced state of the flip operator $\Wopt$ across this bipartition (also called the $\infty-$R\'enyi entropy). The final line \eqref{eq:proof1d} is an exact expression, and can be seen as an alternative version of Thm.~\ref{thm:universal}. Both the min-entropy and the $2-$R\'enyi entropies are valid measures of the LOE.

To arrive at our final result, using Eq.~\eqref{eq:proof1d} we can bound this from both above and below: 
\begin{equation}
   (1 - |\lambda_{\infty}^{(i)} |^2)\leq  D \leq \max_i ( 1 - |\lambda_{\infty}^{(i)}|^{2n}) 
\end{equation}
where the lower bound is valid for any $i$, while the upper bound is for the smallest $\lambda_{\infty}$ across any cut. Then, applying the identity $S^{(\infty)} \leq S^{(2)} \leq 2 S^{(\infty)}$,
\begin{align}
   &1 - \ex^{-S^{(\infty)} (\nu)} \leq D \leq  1 - \ex^{- n S^{(\infty)}(\nu)}  \\
   \iff& 1-\ex^{- \frac{1}{2}S^{(2)}(\nu)}\leq D\leq  1 - \ex^{-n S^{(2)}(\nu)}.
\end{align}
Here, as the left hand side is valid for any $i$, we choose the largest lower bound, corresponding to the largest $2-$R\'enyi entropy across any contiguous bipartition (largest entanglement across any cut). The right hand side is also the largest entropy over bipartitions, so both the upper and lower bound are the same $2-$R\'enyi entropy LOE.
\end{proof}

Let's look at the limits of this bound. Both $S^{(2)}$ and $S^{(\infty)}$ are at most $\log(d/2) = \log(2^{n-1}) \approx n$ (for subsystem of half of total qubits). Then 
\begin{equation}
    1 - \ex^{-n/2} \leq D \leq 1 - \ex^{- n^2},
\end{equation}
and so $D(\Wopt,W_{\text{prod}}) \approx 1$. On the other hand, for $S^{(2)} \approx S^{(\infty)} \approx 0$, we have that $D \approx 0$, which is relevant to the case of Corollary~\ref{cor:clifford}.

Now we will detail the case of when a spoof is approximately strongly universal.
Suppose that instead of finding unitary $\Wopt$ of the form given in Eq.~(\ref{eq:anticomm}), that is, 
\[ \bra{\psi} U_{\theta}^\dagger Z U \ket{\psi} = - \bra{\psi} \Wopt^\dagger U_{\theta}^\dagger Z U_{\theta} \Wopt\ket{\psi}\] 
for all states $\ket{\psi}$, one found a unitary $W$ such that 
\begin{align*}
    \left|\bra{\psi} U_{\theta}^\dagger Z U \ket{\psi} + \bra{\psi} W^\dagger U_{\theta}^\dagger Z U_{\theta} W\ket{\psi}\right| \leq \epsilon
\end{align*}
for all states $\ket{\psi}$.
If $W$ satisfies the above condition, then we call $W$ a $\epsilon$-universal spoof. 
The following results show that if one finds some $W$ such that $\|\Wopt - W\|_\infty \leq \epsilon$ (where again $\|\cdot\|_\infty$ is the Schatten-$\infty$ norm or spectral norm, equivalently the induced $2$-norm), then $W$ is a $2\epsilon$-universal spoof. This is a less strict condition than saying that $\|\Wopt - W\|$ in Schatten 2-norm (Frobenius norm), as the Schatten-$\infty$ norm is bounded by the Schatten 2-norm. Throughout the following proof we leverage the equivalence of the Schatten $\infty$-norm, the spectral norm, and the induced matrix $2$-norm. 

\begin{lemma}
    Let $U, V, W$ be unitary matrices of the same dimension. Suppose $\|U - V\|_\infty \leq \epsilon$. Then for all choices of $W$, 
    \begin{align*}
        \|UWU^\dagger - VWV^\dagger\|_\infty \leq 2\epsilon.
    \end{align*}
\end{lemma}
\begin{proof}
    The assumption that $\|U - V\|_\infty \leq \epsilon$ (and thus $\|U^\dagger - V^\dagger\|_\infty \leq \epsilon$) immediately implies that $\|U^\dagger\ket{a} - V^\dagger \ket{a}\|_{\ell_2} \leq \epsilon$ for all unit vectors $\ket{a}$, where for clarity we use $\|\cdot\|_{\ell_2}$ to denote the Euclidean norm or vector 2-norm. Since $W$ is a unitary matrix and Schatten $p$-norms are invariant under unitary transformation,  $\|WU^\dagger\ket{a} - WV^\dagger \ket{a}\|_\infty \leq \epsilon$.

    Define the (non-unit) vector $\bm{b} := \bm{b}(\ket{a})$ such that $W V^\dagger \ket{a} = WU^\dagger \ket{a} - \bm{b}$, which by definition has (vector) $2$-norm at most $\epsilon$. 
    The definition of $\bm{b}$ implies that    
    \begin{align}
        UWU^\dagger\ket{a} - VWV^\dagger\ket{a} = UWU^\dagger\ket{a} - V(WU^\dagger \ket{a} - \bm{b}) = UWU^\dagger\ket{a} - VWU^\dagger\ket{a} + V\bm{b}.
    \end{align}
    Therefore, for an arbitrary unit vector $\ket{a}$, 
    \begin{align}
        \|UWU^\dagger\ket{a} - VWV^\dagger\ket{a}\|_{\ell_2} = \|UWU^\dagger\ket{a} - VWU^\dagger\ket{a} + V\bm{b}\|_{\ell_2} \leq \|(U-V)WU^\dagger\ket{a}\|_{\ell_2} + \|V\bm{b}\|_{\ell_2} \leq 2\epsilon.
    \end{align}
    This implies that $\|UWU^\dagger - VWV^\dagger\|_\infty \leq 2\epsilon$, where again $\|\cdot\|_\infty$ is the Schatten $\infty$-norm or spectral norm.
\end{proof}

\begin{crllr}
    Let $\Wopt$ be a perfect universal spoof. If $\|\Wopt - W\|_\infty \leq \epsilon$, then $W$ is a $2\epsilon$-universal spoof. 
\end{crllr}
\begin{proof}
    We assume that $\Wopt, W$, and $U_\theta^\dg ZU_\theta$ are unitary. Thus, the above lemma implies that 
    \begin{align}\label{eqn:approx-flip}
        \|\Wopt^\dg U_\theta^\dg ZU_\theta\Wopt - W^\dg U_\theta^\dg ZU_\theta W\|_\infty \leq 2\epsilon.
    \end{align}
    An equivalent characterisation of the spectral norm/Schatten $\infty$-norm is 
    \begin{align*}
        \|A\|_\infty = \max_{\ket{\alpha}, \ket{\beta} \in S^N} |\bra{\alpha}A\ket{\beta}|.
    \end{align*}
    Thus, Eq.~(\ref{eqn:approx-flip}) immediately implies that 
    \begin{align*}
        \max_{\ket{\alpha}, \ket{\beta} \in S^N} \left| \bra{\alpha} \Wopt^\dagger U_\theta^\dagger Z U_\theta \Wopt \ket{\beta} - \bra{\alpha}W^\dagger U_\theta^\dg Z U_\theta W \ket{\beta} \right| \leq 2\epsilon.
    \end{align*}
    Setting $\ket{\alpha} = \ket{\beta} = \ket{\psi(\bm{x})}$ gives the desired result.
\end{proof}

This has ramifications for the effect of the spoof on the states of interest -- as mentioned earlier, $U_{\theta}$ would often be trained until the expectation values for the training data are bounded away from $0$ by a constant. Combined with the above lemma, this suggests that being $\epsilon$-close to a universal spoof, for some small constant $\epsilon$, will result in a large number of the ``relevant'' states being misclassified if the empirical risk does not vanish. 

It is worth noting that similar methods do not seem to work to get a corresponding lower bound (i.e., that a lower bound on $\|U-V\|_\infty$ implies a lower bound on $\|UWU^\dagger - VWV^\dagger\|$). For example, if $W$ commutes with $U$ and $V$, then $\|UWU^\dagger - VWV^\dagger\| = 0$ trivially. Furthermore, in general there may be more than one such ``perfect'' universal spoof, and so one would need to ensure that a given candidate spoof attempt $W$ was bounded away from all valid choices of $\Wopt$. These together suggest that different methods would be needed to determine a corresponding lower bound.

To investigate the operational meaning of the bound of Thm.~\ref{thm:universal}, we conduct numerical simulations that the probe the effectiveness of approximations to universal adversarial attacks by testing the fraction of states misclassified following the attack.
We consider an imperfect universal spoof $W_{\mathrm{approx}}$ satisfying
$\| \Wopt-W_{\mathrm{approx}} \|_2 = \epsilon \sqrt{d} >0$, with $\Wopt$ the perfect spoof and $\epsilon \ll 1$, with $W_{\mathrm{approx}}=e^{-i\epsilon H}\Wopt e^{i\epsilon H}$ for some hermitian $H$ normalised such that $\| H\|_2=\sqrt{d}$.
Then  
\begin{align*}
W_{\mathrm{approx}}^\dg Z W_{\mathrm{approx}}&= \Wopt^\dg \left( \id - i\epsilon H \right)Z \left( \id + i\epsilon H \right)\Wopt\\
&\approx \Wopt^\dg Z \Wopt  - i\epsilon   \Wopt^\dg [H, Z] \Wopt
\end{align*}
A successful spoof (on a given input state $\ket{\psi}$) occurs if $\bra{\psi} W_{\mathrm{approx}}^\dg Z W_{\mathrm{approx}} \ket{\psi}$  has the same sign as $\bra{\psi} \Wopt^\dg Z \Wopt \ket{\psi}$; this therefore happens with probability 
\begin{equation}
    1 - 1/2 \Pr_{\ket{\psi}\sim \mu_{\mbu}}  \left(|\bra{\psi} \Wopt^\dg Z \Wopt \ket{\psi} | < |\bra{\psi} i\epsilon   \Wopt^\dg [H, Z] \Wopt\ket{\psi} |\right) = 1 - 1/2 \Pr_{\ket{\psi}\sim \mu_{\mbu}}  \left(|\bra{\psi}  Z  \ket{\psi} | < |\bra{\psi} i\epsilon  [H, Z] \ket{\psi} |\right)
\end{equation}   
where the factor of one half accounts for the possibility that $|\bra{\psi} i\epsilon    [H, Z] \ket{\psi} |>|\bra{\psi} Z  \ket{\psi} |$ but is of the same sign, in which case the adversarial attack will be successful. Said another way, the adversarial attack is unsuccessful if and only if both $|\bra{\psi} i\epsilon    [H, Z] \ket{\psi} |>|\bra{\psi} Z  \ket{\psi} |$ and $\mathrm{sgn}\left(\bra{\psi} i\epsilon    [H, Z] \ket{\psi} \right)\neq \mathrm{sgn}\left(\bra{\psi} Z  \ket{\psi} \right)$.
As $Z$ and $[H, Z]$ are orthogonal with respect to the Hilbert-Schmidt inner product, $\bra{\psi}  Z  \ket{\psi}$ and $\epsilon \bra{\psi}  [H, Z] \ket{\psi} $ are independent Gaussian variables~\cite{garcia2023deep} when we sample $\ket{\psi}$ according to the Haar measure.
We can work out their mean and variance via the Weingarten calculus~\cite{mele2023introduction} (which together with their Gaussianity tells us the distributions). For the means we have:
\[ \mathbb{E}_{\ket{\psi}\sim \mu_{\mbu}} \bra{\psi}  Z  \ket{\psi} =  \frac{\tr Z}{d} = 0\]
\[ \mathbb{E}_{\ket{\psi}\sim \mu_{\mbu}} \bra{\psi}  \epsilon [H, Z]\ket{\psi} =  \frac{\tr \epsilon [H, Z]}{d}=0 \]
and for the variances:
\begin{align*}
\mathrm{Var}_{\ket{\psi}\sim \mu_{\mbu}} \bra{\psi}  Z  \ket{\psi}&=\mathbb{E}_{\ket{\psi}\sim \mu_{\mbu}} \left( \bra{\psi}  Z  \ket{\psi}^2\right)-\left(\mathbb{E}_{\ket{\psi}\sim \mu_{\mbu}} \bra{\psi}  Z  \ket{\psi}\right)^2\\
&=\mathbb{E}_{\ket{\psi}\sim \mu_{\mbu}} \left( \bra{\psi}  Z  \ket{\psi}^2\right)\\
&=\frac{\tr (Z^2)+\tr (Z)^2}{d(d+1)} \\
&=\frac{1}{d+1}  \\
\end{align*}
Similarly
\[ \mathrm{Var}_{\ket{\psi}\sim \mu_{\mbu}} \bra{\psi} i\epsilon  [H, Z]  \ket{\psi}=   \frac{ \epsilon ^2}{d(d+1)}  \|  [H, Z]\|_2^2      \]
and so 
\begin{align*}
1 - 1/2 \Pr_{\ket{\psi}\sim \mu_{\mbu}}  \left(|\bra{\psi}  Z  \ket{\psi} | < |\bra{\psi} i\epsilon  [H, Z] \ket{\psi} |\right) &= 1 - \frac{1}{\pi} \frac{\sqrt{d}(d+1)}{\epsilon \|  [H, Z]\|_2} \int_0^\infty dy \int_0^y dx \exp \left( \frac{-x^2 (d+1)}{2} \right)\exp \left( \frac{-y^2 d(d+1)}{2\epsilon^2 \|  [H, Z]\|_2^2} \right)\\
&= 1 - \frac{\sqrt{d(d+1)}}{\sqrt{2\pi}\epsilon \|  [H, Z]\|_2} \int_0^\infty dy\ \mathrm{erf}\left(\sqrt{\frac{d+1}{2}}y\right) \exp \left( \frac{-y^2 d (d+1)}{2\epsilon^2 \|  [H, Z]\|_2^2} \right)\\
&= 1 - \frac{\sqrt{d}}{\sqrt{\pi}\epsilon \|  [H, Z]\|_2} \int_0^\infty dy\ \mathrm{erf}\left(y\right) \exp \left( \frac{-y^2 d }{\epsilon^2 \|  [H, Z]\|_2^2} \right)\\
&= 1 -    \frac{1}{ \pi } \arctan  \left( \frac{\epsilon \|  [H, Z]\|_2}{ \sqrt{d} } \right)\\
&\approx 1 -  \frac{\epsilon \|  [H, Z]\|_2}{ \pi\sqrt{d} } \\
&\geq 1 - \frac{2\epsilon }{ \pi } 
\end{align*}
where we have used one of the known closed form expressions for integrals involving products of a Gaussian and the error function~\cite{ng1969table}, as well as the simple bound  \[ \|  [H, Z]\|_2 = \|  HZ-ZH\|_2 \leq \|  HZ\|_2 + \| ZH\|_2 = 2\| H\|_2 = 2\sqrt{d}\] (as $Z$ is unitary). So, for spoof attempts close to the perfect universal attack, the success probability decreases at most only linearly in the 2-norm distance between the approximate and exact attack. In the case $[H, Z]=0$ the success probability does not change at all, and we recover the earlier counterexamples of perfect universal attacks far from a given universal attack in 2-norm.

Finally, we here prove that under unitary evolution according to a Clifford circuit, the LOE is bounded by the number of terms in the Pauli expansion of the initial operator. 

\clifford* 
\begin{proof}
    A flip operator which universally spoofs a prediction according to the measurement of $Z$, as in Eq.~\eqref{eq:prediction}, is 
    \begin{equation}
        F = \sigma_x \otimes \id^{\otimes(n-1)}.
    \end{equation}
    As this is a single Pauli string, the corresponding Choi state is a product state,
    \begin{equation}
        \ket{F} = \sigma_x \otimes \id^{2n-1} \ket{\phi^+}^{\otimes n} =  \ket{\psi^+}  \ket{\phi^+}^{\otimes(n-1)}.
    \end{equation}
This Choi state therefore has zero entanglement, and the LOE of the operator $F$ is zero for any (spatial) bipartition. A unitary circuit $C$ composed of only Clifford gates is itself Clifford, and by definition maps Pauli strings to other Pauli strings. Then it directly follows also that $F_C$ also has zero LOE (across any bipartition):
    \begin{equation}
        S^{(2)} ( \nu ) =0,
    \end{equation}
    for $\nu = \tr_A[C F C^\dg \ket{\phi^+}\bra{\phi^+} C F C^\dg ]  $. For dense angle encoding \eqref{eq:encodingDense}, the adversary effectively has access to any local unitaries on the initial state through classical attacks (of arbitrary strength). Therefore, applying Thm.~\ref{thm:universal}, the upper and lower bounds are tight and so $D=0$. In this case an adversary can implement a universal attack that flips all predictions: $\hat{y}_\theta (\x') = -\hat{y}_\theta (\x)$. 
\end{proof}

As explained in the main text, if $U_\theta$ is a Clifford circuit and $\ket{\phi(\bm{x})}$ is defined according to the dense angle encoding \eqref{eq:encodingDense}, then a universal adversarial attack comprised of local unitaries is guaranteed to exist. This follows from Thm.~\ref{thm:universal} and the property that Clifford circuits map Pauli strings to Pauli strings. Let $\mathcal{F}$ be the set of flip operators that are Pauli strings; there are $2^{2n-1}$ such operators (see Eq.~\eqref{eq:general_flip}). For some particular $F \in \mathcal{F}$, then $U_\theta^\dagger FU_\theta$ ($U_\theta$ Clifford) is a Pauli string, and thus may be applied by a classical adversary that can perform arbitrary local attacks. 

Now we consider the important case of the (non-dense) angle encoding of the form of Eq.~\eqref{eq:encodingAngle}, where the adversary can effectively only induce perturbations of the form $\bigotimes_i(\alpha_i \id + \beta_i Y)$ (thereby including the $2^n$ Pauli strings containing only $\id$ and $Y$). We can apply the following probabilistic argument in the (non-dense) angle encoding \eqref{eq:encodingAngle} case to show that with probability at least $2^{-2^{n}}$ there is a universal counterexample, and that furthermore such a counterexample is easy to find. To do so we assume that that if $U_\theta$ is a 2-design and a Clifford circuit, then $U_\theta P U_\theta^{\dagger}$ is approximately uniformly distributed in $\{I, X, Y, Z\}^{\otimes n}$. 

$P \in \{I, Y\}^{\otimes n}$ is a universal adversarial attack if it is mapped to some element of $\mathcal{F}$ by the action of $U_\theta$. Let $P_i$ be the $i$\textsuperscript{th} Pauli string in $\{I, Y\}^{\otimes n}$ (for some ordering), and let $Q_i$ be the indicator function for the event that $U_\theta P_i U_\theta^\dagger \in \mathcal{F}$. If $Q_j = 0$ for all $j \leq i$ (that is, none of $P_1, \dots, P_i$ map to an element of $\mathcal{F}$), and $U_\theta$ is approximately a 2-design, then 
\begin{align*}
     \Pr(Q_{i+1} = 1 \mid U_\theta P_j U_\theta^\dagger \notin \mathcal{F} ~\forall j \leq i) = \frac{4^{n-1}}{4^{n} - i} \geq \frac{1}{2}.
\end{align*}
Thus, the probability that there is no universal spoof is less than $2^{-2^n}$. This also implies a straightforward process for the classical adversary to find such a spoof, if they have access to $U_\theta$: simply test each $P \in \{I, Y\}^{\otimes n}$ individually until such a spoof is found. The expected number of strings that need to be tested in this case is $O(1)$.

\subsection{Matrix Product State Simulations} \label{ap:mps}
We now describe the simulations of an adversarial attack by an adversary constrained to performing only a tensor product of local operations carried out in Fig.~\ref{fig:numerics}(c). 
We employed a matrix product state (MPS) simulator based on the \texttt{quimb} library~\cite{gray_quimb_2018}. MPS are a tensor network representation of one-dimensional many-body quantum states. In Penrose graphical notation, MPS are represented as~\cite{orus_practical_2014}
\begin{equation}
    \ket{\psi} = \begin{array}{c}
                    \includegraphics[scale=1]{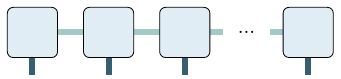}
                \end{array}
\end{equation}
where the shapes are rank-3 tensors (or rank-2 at the boundary) representing each qudit in the system. The internal vertices have a dimension $\chi$ which equals the Schmidt rank of each bipartition and directly relates to the entropy of entanglement via
\begin{equation}
    S_s = -\sum_{i=1}^\chi |\Lambda^{(s)}_i|^2\log(|\Lambda^{(s)}_i|^2) \label{eq:entrop},
\end{equation}
where $s$ indicates the bipartition between the $s$ and $s+1$ qudits and $\Lambda^{(s)}_i$ are the Schmidt values at that bipartition. The maximum entanglement entropy is given by $S=\log\chi$ corresponding to a maximal bond dimension of $\chi=d^{n_{\mathrm{qubits}}/2}$~\cite{perez2006matrix}.

The operator analogues of MPS are the matrix product operators (MPOs) which can be visualised in Penrose notation as
\begin{equation}
    O = \begin{array}{c}
                    \includegraphics[scale=1]{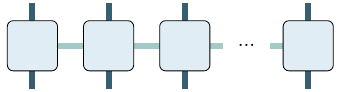}
                \end{array}
\end{equation}
analogously to MPS, the maximum bond dimension of an MPO is given by $d^2$.

Given this, we model the QVC as a series of nearest neighbour 2-qubit unitary operators laid out in a brickwork fashion,
\begin{equation}
    \text{QVC Layer} \sim \begin{array}{c}\includegraphics[scale=2]{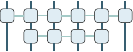}\end{array}
\end{equation}
where the thin internal vertices indicate a bond dimension of one, i.e. a product state. The bond dimension at each bipartition is $2^{2l}$ where $l$ is the number of QVC layers.  To see this consider that an arbitrary 2-qubit unitary is of rank $4$, the operator is then applied across all bipartitions with bond dimension $\chi_{l-1}$ (where $\chi_0=1$). $\chi_l$ is then $4\chi_{l-1}$ as a result of the fusion of the internal vertex of the MPO (with $\chi=\chi_l$) and 2-qubit unitary operator (with $\chi=4$)~\cite{orus_practical_2014}. The bond dimension of the QVC is hence saturated after $n_{\mathrm{qubits}}/2$ layers. The unitary operators that comprise the QVC are generated randomly, with results in Figure \ref{fig:numerics}(c) showing the average of $20$ randomly generated QVCs per layer.

Likewise, the adversary is modelled as a parameterised MPO with all  bond dimensions one to enforce the locality restriction, 
\begin{equation}
    \text{Adversary} \sim \begin{array}{c}
                    \includegraphics[scale=2]{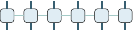}
                \end{array}
\end{equation}
with each tensor having $3$ parameters following the standard construction for a general single-qubit unitary operator $U(\theta, \phi, \lambda)$. As per the discussion at the end of Section~\ref{ap:universal}, in the case of angle encoding, this actually slightly overestimates the capacity of the adversary. Despite this, in Fig.~\ref{fig:numerics}(c), we see the success probability quickly falling as a function of the depth of the circuit.

We utilise the \texttt{quimb}~\cite{gray_quimb_2018} library to perform the simulations of adversarially attacked QVCs. For each value of $n_{\mathrm{qubits}}$, we generate a training set of $32,000$ randomly generated product states with random labels, to which the adversary is trained and a test set of $10,000$ instances. During training, we utilise the cross-entropy loss to determine the likelihood that the adversary flipped the label of the training instance. The training is performed using the \texttt{ADAM}~\cite{kingma2014adam} optimiser with a batch size of $32$ over a single epoch. 

\newpage
\twocolumngrid

\begin{acknowledgments}
M.T.W., N.D., and A.C.N. acknowledge the support of Australian Government Research Training Program Scholarships. N.D. further acknowledges the support of the Monash Graduate Excellence Scholarship. M.U. and M.T.W. acknowledge funding from the Australian Army Research through the Quantum Technology Challenge program. Computational resources were provided by the Pawsey Supercomputing Research Center through the National Computational Merit Allocation Scheme (NCMAS). K.M. acknowledges the support of the Australian Research Council's Discovery Projects DP210100597 and DP220101793. This research is supported by the Ministry of Education, Singapore, under its Academic Research Fund (AcRF) Tier 1 grant, and funded through the SUTD Kickstarter Initiative (SKI 2021\_07\_02).

\end{acknowledgments}

\section*{Data Availability}
The datasets used and/or analysed during the current study are available from the corresponding author upon reasonable request.

\section*{Code Availability}
The code used during the current study is available from the corresponding author upon reasonable request.

\section*{Competing Interests}
The Authors declare no Competing Financial or Non-Financial Interests. 

\section*{Author Contributions}
Dowling and West are Equal contribution authors, listed in alphabetical order. Dowling, West, Modi, Southwell, and Usman conceived the core ideas. Dowling and West co-led the main derivation with technical and numerical support from Nakhl and Southwell. Modi, Usman, and Sevior provided high-level guidance on the technical content. Dowling, West, and Southwell wrote the manuscript with detailed feedback from all authors.

\noindent
{\small \phantom{...} email: ndowling@uni-koeln.de}\\
\noindent
{\small \phantom{...} email: westm2@student.unimelb.edu.au} \\
\noindent
{\small \phantom{...} email: musman@unimelb.edu.au}\\
\noindent
{\small \phantom{...} email: kavan@quantumlah.org}


%



\end{document}